\documentclass[%
 aps,
 amsmath,amssymb,
 reprint,%
]{revtex4-2}

\usepackage{graphicx}% Include figure files
\usepackage{dcolumn}% Align table columns on decimal point
\usepackage{bm}% bold math
\usepackage[utf8]{inputenc}
\usepackage[T1]{fontenc}
\usepackage{mathptmx}
\usepackage{hyperref}
\usepackage{amsmath}
\usepackage{xcolor}

\newcommand*\diff{\mathop{}\!\mathrm{d}}
\newcommand {\ie} {\emph{i.e.,}}
\newcommand {\cf} {\emph{cf.}}

\newcommand {\eq}  [1] {\eqref{eq:#1}}
\newcommand {\Eq}  [1] {Eq.~\eq{#1}}
\newcommand {\Eqs} [1] {Eqs.~\eq{#1}}
\newcommand {\EQ}  [1] {Equation~\eq{#1}}

\newcommand{\fig}[1]{\ref{fig:#1}}
\newcommand{\Fig}[1]{Fig.~\fig{#1}}

\newcommand{\FIG}[1]{Figure~\fig{#1}}

\newtheorem{thm}{Theorem}
\newtheorem{lem}[thm]{Lemma}
\newtheorem{oss}[thm]{Remark}

\begin{document}

\preprint{AIP/123-QED}

\title[\emph{Tayyab}]{Generalized autocorrelation function in the family of deterministic and stochastic anomalous diffusion processes}
% Force line breaks with \\
\author{Muhammad Tayyab} 
\email{muhammad.tayyab@fecid.paf-iast.edu.pk}
\affiliation{School of Computing Sciences, Pak-Austria Fachhochschule Institute of Applied Sciences and Technology, Mang, Haripur, Pakistan 
\\ 
Faculty of Engineering Sciences, Ghulam Ishaq Khan Institute of Engineering Sciences and Technology, Topi, Swabi, Pakistan
}

\date{\today}% It is always \today, today,

\begin{abstract}
%We consider a non-chaotic and deterministic dynamical system of anomalous transport, known as Slicer Map. This uni-dimensional model diffuses in one scaling regime and shows sub-, super-, and normal diffusion under a single parameter variation. Recently, Giberti et al.~\cite{GRTV17} compute analytically $2$-point position correlation function and observe their with numerically estimated correlation functions of the L\'evy-Lorentz gas. We derive a single scaling form of the $3$-point position auto-correlation function of the SM and compare with numerically estimated $3$-point position auto-correlation of the LLg. Here we obtained a remarkable agreement between them, irrespective of any functional relationship of time. Moments of displacement and the position auto-correlations of both system scales in same fashion when all time are large enough and far separated.
We investigate the observables of the one-dimensional model for anomalous transport in semiconductor devices where diffusion arises from scattering at dislocations at fixed random positions, known as L\'evy-Lorentz gas. To gain insight into the microscopic properties of such a stochastically complex system, deterministic dynamics known as the Slicer Map and Fly-and-Die dynamics are used.
%Due to complex stochasticity in the system, direct investigations of such non-trivial dynamics are not possible; therefore, to gain insight into the microscopic properties, we use deterministic dynamics known as the Slicer Map and Fly-and-Die dynamics.
%The Slicer Map is one-dimensional model diffuses in one scaling regime and shows all possible spectrum such as sub-, super-, and normal diffusion under a variation of single parameter.
%%%%%single parameter variation. 
%%%%%%%%%%%%%%%%%%%%%%%%%%%%%%%%%%%%%%%%%%%%%
%We investigate the elusive properties of non-trivial anomalous transport dynamics, called L\'evy-Lorentz gas by using deterministic toy model, known as the Slicer Map. 
%This uni-dimensional map diffuses in one scaling regime and shows all possible spectrum such as
%sub-, super-, and normal diffusion under a 
%%single parameter variation. 
%variation of single parameter.
%%%%%%%%%%%%%%%%%%%%%%%%%%%%%%%%%%%%%%%%%%%
%In order to gain insight on the microscopic properties of much more complex system, the L\'evy-Lorentz gas. 
%%%%%%%%%%%%%%%%%%%%%%%%%%%%%%%%%%%%%%%%%%%%%%%%%%%%%%%%%%%%%%%%%%%%%%%%%%%%%%%%%%%%%%%%%%%%%%%%%%
%From this model [Giberti et. al~\cite{GRTV17}] predict the $2$-point position correlation of the L\'evy-Lorentz gas. 
%%%%%%%%%%%%%%%%%%%%%%%%%%%%%%%%%%%%%%%%%%%%%%%%%%%%%%%%%%%%%%%%%%%%%%%%%%%%%%%%%%%%%%%%%%%%%%%%%%
We analytically derive the generalized position auto-correlation function of these
%the Slicer Map and the Fly-and-Die 
dynamics and study the special case, the $3$-point position correlation function. For this, we derive single parameter-dependent scaling and compare it with the numerically estimated $3$-point position auto-correlation of the L\'evy-Lorentz gas, for which the analytical expression is still an open question. Here we obtained a remarkable agreement between them, irrespective of any functional relationship with time. Moreover, we demonstrate that the position moments and the position auto-correlations of these systems scale in the same fashion, provided the times are large enough and far enough apart. Other observables, such as velocity moments and correlations, are reported to distinguish the systems.
%In order to gain insight on the microscopic properties of much more complex system L\'evy-Lorentz gas, we present conjectured equivalence of the position moments and correlations of the Slicer Map with numerically estimated L\'evy-Lorentz gas correlations. 

\end{abstract}

\maketitle

%\begin{quotation}
%The ``lead paragraph'' is encapsulated with the \LaTeX\ 
%\verb+quotation+ environment and is formatted as a single paragraph before the first section heading. 
%(The \verb+quotation+ environment reverts to its usual meaning after the first sectioning command.) 
%Note that numbered references are allowed in the lead paragraph.
%%
%The lead paragraph will only be found in an article being prepared for the journal \textit{Chaos}.
%\end{quotation}

%\section{\label{sec:level1}First-level heading:\protect\\ The line
%break was forced \lowercase{via} \textbackslash\textbackslash}
%
%This sample document demonstrates proper use of REV\TeX~4.1 (and
%\LaTeXe) in manuscripts prepared for submission to AIP
%journals. Further information can be found in the documentation included in the distribution or available at
%\url{http://authors.aip.org} and in the documentation for 
%REV\TeX~4.1 itself.
%
%When commands are referred to in this example file, they are always
%shown with their required arguments, using normal \TeX{} format. In
%this format, \verb+#1+, \verb+#2+, etc. stand for required
%author-supplied arguments to commands. For example, in
%\verb+\section{#1}+ the \verb+#1+ stands for the title text of the
%author's section heading, and in \verb+\title{#1}+ the \verb+#1+
%stands for the title text of the paper.
%
%Line breaks in section headings at all levels can be introduced using
%\textbackslash\textbackslash. A blank input line tells \TeX\ that the
%paragraph has ended. 

\section{\label{sec:intro}Introduction}
Anomalous transport has been a very active field of research for several decades, but in the past few years it has received enormous attention due to its potential application in numerous fields of science that describe many physical phenomena. For instance, this has been observed in the charge carrier motion in semiconductors~\cite{SM75}, in the polygonal billiards~\cite{JBR08,JeRo06}, in the ion motion within electrolytic cells~\cite{LZRVCMSE17}, in the single molecules inside living cells~\cite{BGM12}, in the ultra-cold atoms~\cite{SBAD12}, in the disordered media~\cite{HA02}, in artificially crowded systems and protein-crowded lipid bi-layer membranes~\cite{JLOM13,JMJM12,SW09}, in the experimental evidence on the mobility of particles in living cancer cells~\cite{GW10} and many others.
%exhibit whole anomalous diffusion.    

The one quantity of interest is to study \emph{transport exponent} $\gamma$ for which the generalized diffusion coefficient 
\begin{equation}
D_\gamma = \lim\limits_{n\rightarrow\infty}\frac{\langle \left(x_n - x_0\right )^2\rangle}{n^\gamma} \in (0,\infty)\,,
\end{equation}
is positive and finite. The numerator ${\langle \left(x_n - x_0\right )^2\rangle}$ represents the \emph{mean square displacement} (MSD) for the position of particle $x_n$ at time $n$. The angular brackets $\langle \cdot \rangle$ correspond to the ensemble average over all particles. The exponent $\gamma$ takes the values $0 \leq\gamma \leq 2$; the transport is called sub-diffusive for $0\leq \gamma<1$, which leads to rapid limit decay; standard diffusion for $\gamma=1$ followed by Fick's law, which has the basic characteristic that the MSD grows linearly in time; super-diffusive when $1< \gamma < 2$, hence the limit diverges, and $\gamma=2$ yields ballistic diffusion. Collectively, except for $\gamma = 1$, this represents a wide spectrum, commonly known as anomalous transport~\cite{FSBZ15,FVGB17,KRS08,WCMS22,MK00,MK04,Kla06}.
%\cite{}.
%\cite{EL18}.%\cite{Kla06,KRS08,MK00,MK04}. 
A contemporary summary of a rich variety of anomalous diffusion processes is provided in~\cite{MJCB14}.
Whereas standard diffusion has been widely investigated in the literature, for instance, see \cite{Gas05,D99,Kla06,Kla96} and references therein.
%Nonetheless, in the era of the deterministic and the stochastic modeling, there are rarely few processes that poses complete anomalous as well as standard behavior in the sense of diffusivity. 
Dynamical systems that exhibit all possible diffusion regimes in the field of anomalous transport are rare in the literature, although in the realm of deterministic dynamics, several authors have investigated anomalous diffusion~\cite{JBS03,JeRo06,JBR08}.
Moreover, in the era of deterministic dynamics, the transport phenomenon is well understood in chaotic systems, which commonly corresponds to standard diffusion. This happened due to the fast rate at which correlation decays. In non-chaotic systems, transport is still underlying, which may often lead to anomalous transport. This is since the rate at which correlation decays is much slower~\cite{Salari,JeRo06,Kla06,KRS08,Zas02,MK04}.
%%%%%%%%%%%%%%%%%%%%%%%%%%%%%%%%%%%%%%%%%%%%%%%%%%%%%%%%%%%%%%%%%%%%%%%%%%%%%%%%%%%%%%%%%%%%%%%%%%
In presence of stochastic elements, the scenario is often
closer to that of chaotic dynamics~\cite{DC00,CFVN02}, but numerous
questions remain open~\cite{GNZ85,BF99,Kla06,Sokolov,MK00,Zas02,Z07,DKU03,LWWZ05,BFK00,BCV10}.
In particular, the asymptotic behaviour of correlation functions is not
understood in general, although it is relevant \emph{e.g.,} to distinguish
transport processes that are effectively different
but have same moments~\cite{Sokolov}. Numerous investigations
have been devoted to this subject, see \emph{e.g.,}~\cite{LWMC23,WCLM20,Z08,BF07,BS07,ZDK15}.
%%%%%%%%%%%%%%%%%%%%%%%%%%%%%%%%%%%%%%%%%%%%%%%%%%%%%%%%%%%%%%%%%%%%%%%%%%%%%%%%%%%%%%%%%%%%%%%%%%

The Slicer Map (SM) was introduced by Salari et al.~\cite{Salari} to study mass transport. The original point of interest was to construct an exactly solvable model (perfect determinism) that would reproduce the transport regimes found numerically in polygonal billiards~\cite{JeRo06}. The SM diffuses in one scaling regime and exhibits sub-, super-, and normal diffusion under a single parameter variation. The position statistics of the SM, including many systems that exhibit strong anomalous transport, are dominated by ballistic trajectories~\cite{VRTGM21,CMGV99,VBB19,VBB20}.
%It has been found that the SM to regenerate the asymptotic scaling of the position moments~\cite{Salari} and the $2$-point position auto-correlation function~\cite{GRTV17} of a more realistic model, the L\'evy-Lorentz gas~\cite{BF99,BCV10} (LLg).
%
It has been proven by Salari et al.~\cite{Salari} that the SM to regenerate the asymptotic scaling of the position moments of a much more complex system, the L\'evy-Lorentz gas (LLG)~\cite{BF99,BCV10}. 
%provided the single parameter $\alpha$ of the SM and the single parameter $\xi$ of the LLG are properly tuned. 
The LLG is a one-dimensional random walk in a random environment in which the scatterers are randomly distributed on a line according to a L\'evy-stable probability distribution. Burioni et al.~\cite{BCV10} used simplifying assumptions to determine the mean-square displacement of the traveled distance and numerically validated their findings. More recently, Bianchi et al.~\cite{BCLL16} presented a rigorous mathematical study to prove a central limit theorem, and M. Zamparo~\cite{Z23} investigated large fluctuations and transport properties of the LLG.
It was further proven by Giberti et al.~\cite{GRTV17} when the single parameters $\alpha$ and $\xi$ of the SM and the LLG, respectively, are properly tuned, this choice of parameters leads, within the transport regimes of the LLG, to the equality of the asymptotic scalings of the $2$-point position auto-correlation functions. This makes two very different systems indistinguishable regardless of their microscopic dynamics, as far as the statistics of positions are concerned. Indeed, such an agreement does not infer a full equivalence of the dynamics. For example, the trajectories of the SM move ballistically in an initial transit and then turn periodic in a period-two cycle, then remain oscillating back and forth between their neighbouring cells, while in the LLG all trajectories are stochastic. This fact is further addressed in sec.~\ref{subsec:SM} and \ref{sec:concl}.
%on the level of the statistics of positions.
%and the $2$-point position auto-correlation function~\cite{GRTV17} of a more realistic model, the L\'evy-Lorentz gas~\cite{BF99,BCV10} (LLg).
%Salari et al.~\cite{Salari} shows that once $\alpha$ and $\beta$ are properly tune then all the position moments of the SM coincides with those of LLg. In consistency, Giberti et al.~\cite{GRTV17} compute the position auto-correlation function of the SM and compared with numerically computed PACF of the LLg, therefore the PACF of these system also coincide atleast for small values of $\beta$.
%%%%%%%%%%%%%%%%%%%%%%%%%%%%%%%%%%%%%%%%%%%%%%%%%%%%%%%%%%%%%%%%%%%%%%%%%%%%%%%%%%%%%%%%%%%%%%%%%%%%

The deterministic and time-continuous prototypical model, Fly-and-Die (FND) dynamics, was introduced to mimic the universal features of displacement statistics~\cite{VRTGM21}. The FND dynamics exhibit a wide spectrum of diffusion, from sub-, normal, and super, upon varying a single parameter. In the FND, strong anomalous diffusion (super diffusion) emerges due to ballistic trajectories, \ie the ballistic trajectories that did not undergo transitions up to any finite time (see for instance \cite{GSSPCM18}). It is further motivated by the fact that subdominant terms in the FND and the SM contribute like ballistic flights to the asymptotic behaviour, \ie~they contribute the maximum allowed for a system to belong to the universal behaviour. This is further proven by the fact that, analytically, all the position moments and the $2$-point position correlations of the FND asymptotically scale as those of the SM despite having different microscopic structures. Upon tuning the diffusion parameter of the SM and the FND with the LLG in accordance with this agreement, all the moments coincide analytically and with the remarkable numerical agreement. At the same time, the $2$-point position correlation exhibits the same power law behavior as those numerically estimated position correlations in the LLG \cite{GRTV17,VRTGM21}.

%%%%%%%%%%%%%%%%%%%%%%%%%%%%%%%%%%%%%%%%%%%%%%%%%%%%%%%%%%%%%%%%%%%%%%%%%%%%%%%%%%%%%%%%%%%%%%%%%%%%

In this paper, we intend to explore the equivalence of the higher-order position auto-correlation function of the SM, the FND, and the LLG and see up to which order of correlation the SM and the FND are indistinguishable from the LLG. For this, first, we derive the generalized position auto-correlation function of the SM and the FND, and then, for the particular case, \ie the $3$-point position auto-correlation function, we compute a single scaling form that depends only on one parameter: $h_2/n_2$ (\cf~section~\ref{subsec:3pcorr} and \Eq{U_3pt_corr_FnD}), which is simply the ratio of times.
%The velocity moments and correlation function are also reported to observe the dissimilarities in these systems. 
%For this, first we derive the $3$-point position auto-correlation function of the SM in a single scaling form that depends only one parameter: $\tau_2/(t_1+\tau_1)$ (\cf~section~\ref{subsec:3pcorr}), which is simply ratio of times.
We compare one parameter-dependent analytical scaling form of correlation with the numerically estimated position auto-correlation function of the much more realistic model, the LLG. We find remarkable agreement between the correlation scaling of the SM, the FND, and all numerically estimated $3$-point position correlation of the LLG. Regardless of any functional relationship between the time, all data sit on top of each other and have a nice agreement with the theoretical prediction (\cf~\Eqs{U_3pt_corr} and \eq{U_3pt_corr_FnD}). Moreover, the velocity moments and correlation function are also reported to observe the dissimilarities in these systems.
%we derive the velocity moments and correlations of the SM and the FnD to distinguish the systems.%
On the contrary, we also argue about the statistics of the position moments and the correlations that scale in the same way, see Fig. \ref{fig:expnt-data}. This is due to the fact that, in the correlation function, separation between different times becomes irrelevant as compared to the mean.

This paper is organized as follows: Section~\ref{subsec:SM} formally summarises the SM and illustrates its properties. Section~\ref{sec:SM-npoint-corr} provides the $m$-point position auto-correlation function expression. Section~\ref{subsec:3pcorr} demonstrates a scaling formula for the $3$-point position auto-correlation function. In Sections~\ref{sec:Moments-v-SM} and \ref{sec:Corr-v-SM}, we formally introduce velocity moments and correlation functions of the SM. Section~\ref{sec:FnD-model} does the same for the FND dynamics. Section~\ref{sec:LLg} devotes itself to the LLG, which characterises the properties of the system and also reports numerical results on the applicability of the scaling formula for the $3$-point position auto-correlations. Section~\ref{sec:concl} summarises our conclusions.

\section{Deterministic dynamics\label{sec:det}}
%%%%%%%%%%%%%%%%%%%%%%%%%%%%%%%%%%%%%%%%%%%%%%%%%%%%%%%%%%%%%%%%%%%%%%%%%%%%%%%%%%%%%%%%%%%%%%%%%%%%
\subsection{The slicer map\label{subsec:SM}}
The SM is one dimensional, deterministic and exactly solvable dynamics \cite{GRTV17,Salari}. Its time evolution is given by the map
\begin{align*}
S_\alpha : 
 \left[ 0, 1 \right] \times \mathbb{Z} \to 
\left[ 0, 1 \right] \times \mathbb{Z}
% \left[ 0, 1 \right] \times \mathbb{Z} \, ,
 \end{align*}
\begin{subequations}\label{eq:Sys-SM}
defined by
%%%%%%%%%%%%%%%%%%%%%%%%%%%%%%%%%%%%%%%%%
\begin{align}\label{eq:SM-EOM}
&x_{n+1} 
 =
 S_\alpha(x_n) \nonumber   \\
  &= \left\{
    \begin{array}{rl}
      (x_n,m-1) & \mbox{ if \quad} 0\leq x_n \le  \ell_{m} \mbox{ \;or\; } \frac{1}{2} < x_n \leq 1-\ell_{m},\\[3mm]
      (x_n,m+1) & \mbox{ if \quad} \ell_{m} < x_n \leq \frac{1}{2} \mbox{ \;or\; } 1-\ell_{m} < x_n \leq 1 \, ,
    \end{array}
                \right.
\end{align}
here 
\( x_{n} = \{ x + n \} \), with \( x \) is the fractional part (\ie~\( 0 \leq x < 1 \)) and \( n \in \mathbb{N}_0 \) a non-negative integer. In each term \( x_{n} \), \( n \) is added to \( x \), and the use of the fractional part function ensures that \( x_{n} \) remains within the interval \([0,1]\). The initial ensemble $x_0$ is chosen uniformly in the interval $[0,1]$. 

%%%%%%%%%%%%%%%%%%%%%%%%%%%%%%%%%%%%%%%%%
The family of slicers 
\begin{equation}
\ell_{m} =\frac{1}{(|m|+2^{1/\alpha})^\alpha}, \quad\text{ with } \alpha \in \mathbb{R}^+ \,,
\end{equation}
\end{subequations}
determines the position of the slicer and chops the slices in their neighboring cells.

For $1/2 < x_n < 1$ each iteration of the map increases
the values of $m$ by one, until $x_n > \ell_{m}$. Subsequently,
the trajectory enters a stable period-two cycle, oscillating back and forth between the two neighboring sites $m$ and $m - 1$. Similarly, for $0 < x_n < 1/2$ each iteration
of the map decreases the values of $m$ by one until $x_n < -\ell_{m}$, and then the trajectory enters a stable period two-cycle. 
The distance between two trajectories does
not change in time, as long as they are mapped by the same branch of the map which, for each $m\in\mathbb{Z}$, are defined by the \emph{``slicer''} $\ell_{m}$. 
The distance between two points
$x_n$ and $x_{n+1}$ jumps discontinuously when they reach a cell $m$ where $\ell_{m} \in [x_n, x_{n+1}]$. Thus, the dynamics is reminiscent
of polygonal billiard dynamics \cite{JBR08,JeRo06}, where initial conditions are only separated when they are reflected by different sides of the
polygon. The corners act as slicers of the bundle of initial
conditions. The analogy between the two systems
also includes the fact the SM has a vanishing Lyapunov exponent and it preserves the phase space volume. Like the SM exhibits sub-, super-, and normal diffusion upon varying the parameter $\alpha$ that describes the position of the slicers (see~\cite{Salari,GRTV17} and Ref. therein). Therefore such deterministic dynamics are rare in the literature of transport processes that shows a wide spectrum of diffusion.
\begin{figure}
\hspace{-6.5mm}\includegraphics[width=3.62in]{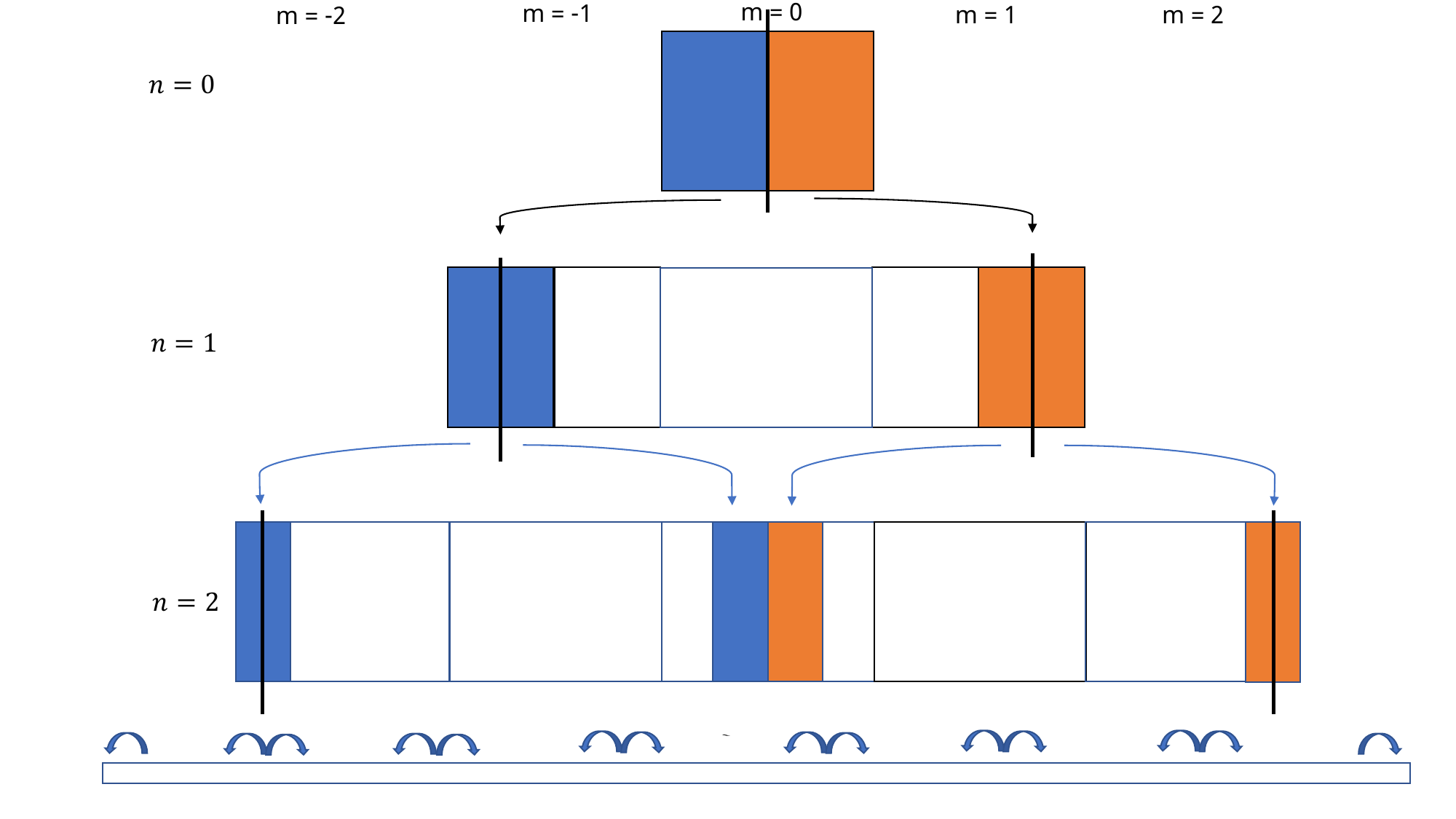}% Here is how to import EPS art
\caption{\label{fig:SM-Map} Demonstration of space-time plot for the Slicer Map defined by $S_\alpha$ in \Eq{Sys-SM}, where $n$ represents time and $m$ space; shown is the diffusive spreading of points that at $n=0$ are uniformly distributed initial condition on the unit interval centered around $m=0$. Iteration of the map $S_\alpha$ is shown up to time $n=2$, it represents that points in the centered cell $m=0$ start moving in their neighbouring cells as $n$ grows. The lower horizontal strip with back and forth signs on the central and forward-backward signs on the edges denotes sub-travelling and travelling points respectively as $n$ grows.}
\end{figure}
%%%%%%%%%%%%%%%%%%%%%%%%%%%%%%%%%%%%%%%%%%%%%%%%%%%%%%%%%%%%%%%%%%%%%%%%%%%%%%%%%%%%%%%%%%%%%%%%%%
%is small enough to fit in a single column, while
%Fig.~\ref{fig:wide}%
%\begin{figure*}
%\includegraphics{fig_2}% Here is how to import EPS art
%\caption{\label{fig:wide}Use the \texttt{figure*} environment to get a wide
%figure, spanning the page in \texttt{twocolumn} formatting.}
%\end{figure*}
%\paragraph{Fourth-level heading is run in.}%
%Footnotes are produced using the \verb+\footnote{#1}+ command. 
%Numerical style citations put footnotes into the 
%bibliography\footnote{Automatically placing footnotes into the bibliography requires using BibTeX to compile the bibliography.}.
%Author-year and numerical author-year citation styles (each for its own reason) cannot use this method. 
%Note: due to the method used to place footnotes in the bibliography, \emph{you
%must re-run BibTeX every time you change any of your document's
%footnotes}. 

\subsubsection{$p^{\text{th}}$ position moments}
Salari et al.~\cite{Salari} introduced the SM and calculated all moments of the displacement as a function of the number $n$ of iterations of the map. In the following, we review
those calculations differently, but for the sake of simplest representation, we shift the origin of the positions by $1/2$, so that the right half of the unit interval
coincides with $[0, 1/2]$, rather than $[1/2, 1]$. This does not
affect the asymptotic results that are obtained from an ensemble of initial conditions with $m = 0$ and $x_0$ uniformly
distributed in the right half of the unit interval. For $n\gg 2^{1/\alpha}$
the $p^\text{th}$ position moments amounts in following Lemma~\ref{lem:pmoments}.
%Salari et al.~\cite{Salari} calculated all the displacement moments as function of the number $n$ of iterations of the map, while Giberti et al.~\cite{GRTV17} derived the moments of displacement in alternative fashion from \cite{Salari} and compute $2$-point PACF and present their equivalence with LLg correlations. 
%Here we revisit the moments of displacement as a simplest representation class of the SM. For $n\gg 2^{1/\alpha}$ the $p^\text{th}$ moments amounts to
\begin{lem}\label{lem:pmoments}
Given $\alpha>0$, the $p^{th}$ position moments of the Slicer dynamics for uniformly distributed initial condition, asymptotically scales as
%\begin{subequations}
\begin{eqnarray}
  \left\langle (x_n - x_0)^p \right\rangle  
 %\simeq & 2 \,\int\limits_0^{\ell_n} \diff x \, n^p       + 2 \,\int\limits_{\ell_n}^{1/2} \diff x \, \left( x^{-1/\alpha} - 2^{1/\alpha} \right)^p
%  \nonumber\\[3mm]
%  & \sim 2\, n^p  \:  \ell_n 
%    + \frac{2}{1 - p/\alpha} \: ( 2^{-1+p/\alpha} - \ell_n^{1-p/\alpha} ) 
%    + \mathcal{O}(1)
%  \nonumber\\
%  & \sim
%    \frac{2 \,p}{p - \alpha} \: n^{p - \alpha} + \mathcal{O}(1)
%  \\[2mm]
\sim \left\{ \begin{array}{ll}\label{eq:SM-moments}
const.\,,       & \text{ for }  p < \alpha \, ,
                   \\[2mm]
                   \displaystyle 2 \; \ln\frac{n^\alpha}{2}\, ,         & \text{ for }  p = \alpha \, ,
                   \\[2mm]
                   \displaystyle \frac{2\, p}{p-\alpha} \, n^{p-\alpha}\, ,  & \text{ for }  p > \alpha > 0 \, .
    \end{array}\right.
\end{eqnarray}
%\end{subequations}
\end{lem}
\begin{proof}
See appendix~\ref{app:pth-momnt}.
\end{proof}

\noindent For $p=2$, the \emph{MSD} $\left\langle (x_n - x_0)^2 \right\rangle \sim n^\gamma$ where $\gamma = {2-\alpha}$ with $0<\gamma<2$, captures all scenarios of anomalous diffusion. The SM exhibits super-diffusion for $\gamma > 1$; for $\gamma = 1$, the power law grows linearly in time, \ie normal diffusive; and for $\gamma < 1$, it is sub-diffusive. Since there is no drift in the SM, all odd moments vanish. One confined the motion of particles in one direction, we can also identify all odd position moments. 

%while $p=\alpha$ leads to
%\begin{equation}
%\left\langle (x_n - x_0)^\alpha \right\rangle
%  \sim                  \displaystyle 2 \; \ln\frac{n^\alpha}{2}  \, .
%\end{equation}
%These asymptotic scaling coincide with the moments \Eq{FnD-moment-scaling} of the FnD dynamics when one consider $\alpha=2-\gamma$ and $b\equiv 2$, except $p<\alpha$.

%It has been observed by Giberti et al.\cite{GRTV17}, 
The $2$-point position auto-correlation function of the SM asymptotically scales as those of the numerically estimated position correlations of the LLG~\cite{GRTV17}.
%Giberti et al.~\cite{GRTV17} computed the $2$-point position auto-correlation function of the SM and present their conjectured equivalence with the position correlations of the LLg. 
Since our fundamental objective is to observe the equivalence of the position auto-correlation functions of different dynamics, only partial equivalence at the level of all position moments and the correlation function of order $2$ do not suffice to determine the indistinguishability of the dynamics. This may leave many unanswered questions; \emph{e.g.,} as far as statistics of positions are concerned, one does not know to what extent the SM, the FND, and the LLG are indistinguishable.
 % how far the SM captures the LLg.
% correlations.
%to study elusive microscopic properties of the dynamics. 
%Indeed, findings on the equivalence of the position moments and $2$-point correlation functions are not enough to indistinguish the systems.
To address these questions, in the following, we explicitly derive the generalized 
%(or $m$-point) 
position auto-correlation function of the SM and the FND to see how far equivalence holds to the numerically estimated correlations of the LLG.
%see how far equivalence hold to correlation of the LLg which are estimated numerically.
%numerically estimated correlations of the LLg.
%the LLg position correlations.
% of the LLg.
%%%%%%%%%%%%%%%%%%%%%%%%%%%%%%%%%%%%%%%%%%%%%%%%%%%%%%%%%%%%%%%%%%%%%%%%%%%%%%%%%%%%%%%%%%%%%%%%%%
\subsubsection{Generalized position auto-correlation function}\label{sec:SM-npoint-corr}
The generalized (or $m$-point) position auto-correlation function  of the SM for time
%$n_1\leq n_2\leq \cdots \leq n_m $ 
$n_m \geq n_{m-1}\geq\cdots\geq n_2\geq n_1,$
is defined as
\begin{align}
&\phi_\alpha(n_1,n_2,\cdots,n_m)=\langle \, (x_{n_1}-x_0)\,\,\cdots\, (x_{n_{m-1}}-x_0)\,(x_{n_m}-x_0)\, \rangle
\nonumber\\
&= \langle \Delta x_{n_1}\, 
%%%%%%\Delta x_{n_2} \,
\cdots\, \Delta x_{n_{m-1}}\,\Delta x_{n_m} \rangle\,  
 = \int\limits_0^{1/2} \diff x\; \Delta x_{n_1}\, %%%%%%%\Delta x_{n_2} \,
 \cdots\, \Delta x_{n_{m-1}}\,\Delta x_{n_m}\,.
\end{align}
According to the flight of trajectories, the integration interval $\mathcal{F}=(0,\, 1/2]$ is partitioned into $m+1$ parts, \ie~$\mathcal{F} =
 {L^{>\,n_m}} 
\cup {L^{>\,n_{m-1}}} \: \cup\: \,\cdots \, \: \cup\: {L^{>\,n_1}}\:\cup \:{L^{\geq\,1/2}}$, defined as
\begin{align}
&\phi_\alpha(n_1,n_2,\cdots, n_m) \nonumber\\
  &= \int\limits_{L^{>\,n_m}} \diff x\; \Delta x_{n_1}\, \Delta x_{n_2} \,\cdots\, \Delta x_{n_m}\,
+ 
\int\limits_{L^{>\,n_{m-1}}}  \diff x \;\Delta x_{n_1}\, \Delta x_{n_2} \,\cdots\, \Delta x_{n_m}+ 
 \\
&\cdots 
+
\int\limits_{L^{>\,n_1}}  \diff x\; \Delta x_{n_1}\, \Delta x_{n_2} \,\cdots\, \Delta x_{n_m}\,
+
\int\limits_{L^{\geq\,1/2}}  \diff x\; \Delta x_{n_1}\, \Delta x_{n_2} \,\cdots\, \Delta x_{n_m}\,\nonumber\,,  
\end{align}
with $n_m \geq n_{m-1}\geq\cdots\geq n_2\geq n_1\,.$

The integration limits are separated accordingly to their trajectory flying time
\begin{description}
\item[${L^{>\,n_m}} = \{0<x<\ell_{n_m}\}$] 
						All trajectories are flying at all times, such that $\Delta x_{n_k}=n_k\,,\,$ with $k=1,\,2,\,3,\,\cdots,\, m\,.$
%\Delta x_{n_2}=n_2 \,$ and  $ \Delta x_{n_1}=n_1\,$.
\item[${L^{>\,n_{m-1}}} = \{\ell_{n_m}<x<\ell_{n_{m-1}}\}$] 
						The trajectory is still flying at time $n_{m-1}$, but it has localized (turned periodic) by time $n_m$, consequently 
						$\Delta x_{n_m}=\left(x^{-1/\alpha} - 2^{1/\alpha}\right)\,$ 
						and 
						$\Delta x_{n_k} = \prod\limits_{k=1}^{m-1} n_k\,$.
						%with $k=1,\,2,\,3,\,\cdots, \,m-1\,.$		
			\\	\vdots 
\item[${L^{>\,n_1}} = \{\ell_{n_2}<x<\ell_{n_1}\}$] 
						The trajectory is still flying at time $n_1$, but it has localized by time $n_2$ and subsequent, consequently $\Delta x_{n_1} = n_1,$
and						 $\Delta x_{n_k}=\left(x^{-1/\alpha} - 2^{1/\alpha}\right)^k\,$ with $k=2,\,3,\,\cdots, \,m-1\,.$
\item[${L^{\geq\,1/2}} = \{\ell_{n_1}<x<1/2\}$] 
						All trajectories get localized before time $n_1$, hence $\Delta x_{n_k} = \left(x^{-1/\alpha} - 2^{1/\alpha}\right)^k\,$ with $k=1,\,2,\,3,\,\cdots,\, m\,.$
\end{description}
Therefore for $n_m \geq n_{m-1}\geq\cdots\geq n_2\geq n_1\,$, integrals emerges as
%\begin{subequations}\label{eq:N-point-SM}
\begin{widetext}
\begin{align}
&\phi_\alpha(n_1,\,n_2,\,\cdots,\,n_m)
  \nonumber \\ 
& \simeq 2 \left(n_1 n_2 \cdots n_m\right)\int\limits_{0}^{\ell_{n_m}}\diff x  
+ 2\left(n_1 n_2 \cdots n_{m-1}\right)\int\limits_{\ell_{n_m}}^{\ell_{n_{m-1}}}\diff x \,\left(x^{-\frac{1}{\alpha}} - 2^{\frac{1}{\alpha}}\right)
%+
%2\,(n_1 n_2 \cdots n_{m-2})\int\limits_{\ell_{n_{m-3}}}%^{\ell_{n_{m-2}}}\diff x (x^{-\frac{1}{\alpha}} - 2^{1/%\alpha})^2
+ \cdots +
2\,n_1\int\limits_{\ell_{n_{2}}}^{\ell_{n_{1}}}\diff x \left(x^{-\frac{1}{\alpha}} - 2^{\frac{1}{\alpha}}\right)^{m-1}
+
2\int\limits_{\ell_{n_{1}}}^{1/2}\diff x \left(x^{-\frac{1}{\alpha}} - 2^{\frac{1}{\alpha}}\right)^{m},\nonumber\\
&
\simeq 2 \sum\limits_{j=0}^{m} 
\left( 
\prod\limits_{k=1}^{m-j} n_k 
\int\limits_{\ell_{n_{m-j+1}}}^{\ell_{n_{m-j}}}
\diff x \left(x^{-\frac{1}{\alpha}} - 2^{\frac{1}{\alpha}}\right)^{j} 
\right)
\sim
%\stackrel{x^{-1/\alpha}  >  2^{1/\alpha}}{=}
2\sum\limits_{j=0}^{m} 
\left[ 
\prod\limits_{k=1}^{m-j} n_k 
\left( \frac{\alpha}{\alpha-j}  \left(n^{j-\alpha}_{{m-j}} - n^{j-\alpha}_{{m-j+1}} \right)  \right)
\right],\label{eq:N-point-SM-summ}
\end{align}
\end{widetext}
%\end{subequations}
%and a simple series form yields the following expression for the $m$-point PACF:
and $n_{m+1}=0$ and $n_0=K$, where $K$ is constant. 
%In following we compute the $3$-point position auto-correlation function by adopting $m=3$, in \Eq{N-point-SM-summ}.
%%%%%%%%%%%%%%%%%%%%%%%%%%%%%%%%%%%%%%%%%%%%%%%%%%%%%%%%%%%%%%%%%%%%%%%%%%%%%%%%%%%%%%%%%%%%%%%%%%
\subsubsection{$3$-point position auto-correlation function\label{subsec:3pcorr}}
The $3$-point position auto-correlation function can be obtained by requesting $m=3$ in \Eq{N-point-SM-summ}, 
%and adopting $\alpha=2-\gamma$, 
the correlation function amounts to
%\begin{align}
% \hspace{-2mm}\phi(n_1,\,n_2,\,n_3)   &\simeq \frac{2\; n_{1}\,n_{2}\,n_{3}^{1-\alpha}}{1-\alpha}
%   + 
%   \frac{2\,\alpha\;n_{1}\,n_{2}^{2-\alpha}}{(1-\alpha)(\alpha-2)} 
%   + 
%   \frac{2\,\alpha\;\,n_{1}^{3-\alpha}}{(2-\alpha)(\alpha-3)} \,,
%\nonumber   \\[2.5mm]   
%  \quad& + \mathcal{O}(n_1\,n_2^{2-\alpha}, n_1\,n_2^{1-\alpha} ,n_1\,n_2^{-\alpha}, n_1^{1-\alpha}, n_1\,n_2\,n_3^{-\alpha},1)\,.
% % \qquad \text{with} \quad n_1\geq n_2 \geq n_3\,.
% \nonumber
%       \end{align}
%One can take $\alpha=2-\gamma$, and write the correlation expression as 
%\begin{align}\label{eq:3p-SM-U-Scal}
%\phi_\alpha(n_1,\,n_2,\,n_3) \simeq
%\frac{2\,n_1\,n_2\,n_3^{\gamma-1}}{\gamma-1} - \frac{2(2-\gamma)\,n_1\,n_2^{\gamma}}{\gamma(\gamma-1)} 
%- \frac{2(2-\gamma)\,n_1^{\gamma+1}}{\gamma(\gamma+1)}\,.
%%\qquad \gamma\neq 1\,.
%\end{align} 
%%%%%%%%%%%%%%%%%%%%%%%%%%%%%%%%%%%%%%%%%%%%%%%%%%%%%%%%%%%%%%%%%%%%%%%%%%%%%%%%%%%%%%%%%%%%%%%%%%
\begin{equation}
  \phi_\alpha(n_1, n_2, n_3) 
  \simeq \left\{ \begin{array}{ll}
  \frac{2\,n_1\,n_2\,n_3^{1-\alpha}}{1-\alpha} - \frac{2\alpha\,n_1\,n_2^{2-\alpha}}{(2-\alpha)(1-\alpha)} 
- \frac{2\,\alpha\,n_1^{3-\alpha}}{(2-\alpha)(3-\alpha)}\,,
  &   \alpha \neq 1\,,
  \\[4mm]
 2\,n_2\,n_3 (\ln \frac{n_1}{n_2}+2) + 24\,\ln\frac{n_3}{2}- n_3^2
  \\
  + 8\,n_3  (\ln \frac{n_3}{n_2}-2)
 \,,
  &   \alpha = 1\,.
  \end{array}
\right .\label{eq:3p-SM-U-Scal}
\end{equation}
%%%%%%%%%%%%%%%%%%%%%%%%%%%%%%%%%%%%%%%%%%%%%%%%%%%%%%%%%%%%%%%%%%%%%%%%%%%%%%%%%%%%%%%%%%%%%%%%%%
\begin{itemize}
  \item[I.]
		 For any fixed time $n_1$ and for the other two equivalent times $n_2=n_3$, we recover the asymptotic scaling of the \emph{MSD} (\cf~\Eq{SM-moments} with $p=2$) as given by
\begin{equation}\label{eq:2-MOMNT-scaln}
\langle (x_n-x_0)^2 \rangle \sim \frac{4}{2-\alpha}n^{2-\alpha}\,, \quad 0<\alpha< 2\,.
\end{equation}
 \item[II.]
 		For $n_1=n_2=n_3$, this reduces to the third moment of displacement (\cf~\Eq{SM-moments} with $p=3$),
such as
\begin{equation}\label{eq:3-MOMNT-scaln}
\langle (x_n-x_0)^3 \rangle \sim \frac{6}{3-\alpha}n^{3-\alpha}\,, \quad 0<\alpha< 3\,.
\end{equation}
\end{itemize}
%\paragraph{For $n_1=n_2=n_3$, this reduces to the third moment of displacement (cf.~\Eq{SM-moments} with $p=3$),
%%and $\gamma = 2-\alpha$)\,, 
%such as
%\begin{equation}\label{eq:3-MOMNT-scaln}
%\langle (x_n-x_0)^3 \rangle \sim \frac{6}{3-\alpha}n^{3-\alpha}\,, \quad 0<\alpha< 3\,.
%\end{equation}
%}
%\paragraph{For $n_1$ is fixed and $n_2=n_3$, we recover the asymptotic scaling of the MSD $\sim n^{2-\alpha}\,,\; where\;\, 0<\alpha <2\,.$\\} 
A few cases of the time-composition can be defined as follows
\begin{enumerate}  
\item  $h_1 = n_2 - n_1$, \; as \;$h_1>0$, either finite or $h_1\sim n_1^{q_1},\;\; q_1\leq 1$, and $n_1\rightarrow \infty$\,.
\item $h_2= n_3 - n_2$, \; as \; $h_2>0$ either finite or $h_2\sim n_1^{q_2},\;\; q_2\leq 1,$ 
where $n_2 = n_1+h_1$, and $n_1 \rightarrow \infty$\,.
\item $n_1\geq n_2$ are fixed and set $n_3\rightarrow\infty$\,.
\end{enumerate}
If one sets all $n$ times tend to infinity for the $3$-point position auto-correlation function $\phi_\alpha(n_1,n_2,n_3)$, as given in \Eq{3p-SM-U-Scal}, the power law exponent scales in the same way as found for the third position moment (see Lemma~\ref{lem:nptscaling-SM} for $m=3$). To address this, we consider the scaling of the correlation $\phi_\alpha(n_1,n_1+h_1,n_1+h_1+h_2)$ for very large values of $n_1$. For large $n_1$, when the time lags $h_{k-1}$ are either constant or scale as $\sim n^q$ for $k=2,3$ with $q<1$, the difference among the three times becomes negligible compared to the mean.
\begin{lem}\label{lem:nptscaling-SM}
For $0<\alpha <m$, as all $n$ times tend to infinity
%$n_k\rightarrow \infty$, 
and 
$h_{k-1} = n_k-n_{k-1}$, for $k=2,3, \dots, m$, where $h_{k-1}$ either fixed or $\sim n_1^q,\;q<1$, the $m$-point position auto-correlation function $\phi_\alpha$ represented in \Eq{3p-SM-U-Scal}, asymptotically scales as
\begin{equation}\label{eq:nptCORR-scal}
\phi_\alpha(n_1,n_2,\cdots, n_m)\sim \frac{2\,m}{m-\alpha} n_1^{m-\alpha}\,, \quad 0<\alpha< m\,.
\end{equation}
\end{lem}
\begin{proof}
This is a direct consequence of \Eq{3p-SM-U-Scal}.\hfill$\square$
\end{proof}
%%%%%%%%%%%%%%%%%%%%%%%%%%%%%%%%%%%%%%%%%%%%%%%%%
%\begin{lem}\label{lem:3ptscaling}
%For $n_1\rightarrow \infty$ with $n_2-n_1=h_1$ and $n_3-n_2=h_2$, where $h_1$ and $h_2$ either fixed or $\sim n_1^q,\;q<1$, the $3$-point position auto-correlation function $\phi_\alpha$ represented in \Eq{3p-SM-U-Scal}, asymptotically scales as
%\begin{equation}\label{eq:3ptCORR-scal}
%\phi_\alpha(n_1,n_2,n_3)\sim \frac{6}{3-\alpha} n_1^{3-\alpha}\,, \quad 0<\alpha< 3\,.
%\end{equation}
%%where time-composition adopted as in Sec.~\ref{subsec:3pcorr}.
%\end{lem}
%\begin{proof}
%This is a direct consequence of \Eq{3p-SM-U-Scal}.\hfill$\square$
%\end{proof}
%%%%%%%%%%%%%%%%%%%%%%%%%%%%%%%%%%%%%%%%%%%%%%%%%
\begin{oss}\label{rem:mom-corr-SM}
For $0<\alpha<m$, where $m=2$ or $3$, as $n$ tends to infinity, the $2$- and $3$point position auto-correlation function $\phi(\alpha)$, following the $h_{k-1}$ represented in Lemma~\ref{lem:nptscaling-SM}, exhibit the same asymptotic scaling as the second moment, i.e., MSD (see \Eq{nptCORR-scal}), and the third moment of displacement (see \Eq{3-MOMNT-scaln}), which is
\begin{equation}\label{eq:sam-sclng-momnt-corr-SM}
\phi_\alpha(n_1,n_2) \sim \langle (x_n-x_0)^2\rangle, \;
%\;\; \text{and} \;\;
\phi_\alpha(n_1,n_2,n_3) \sim \langle (x_n-x_0)^3\rangle,
\end{equation}
respectively.
\end{oss}
For the single parameter correlation function, reconsider \Eq{3p-SM-U-Scal}, rearrange some terms, introduce the time lag $h_2=n_3-n_2$ and $h_1=n_2-n_1$ and the normalization factor $n_1\,n_2^{2-\alpha}$, we find
\begin{subequations}\label{eq:U_3pt_corr}
\begin{widetext}
\begin{align}
\phi_\alpha\left(\frac{h_2}{n_2}\right)= \frac{\phi_\alpha(n_1,\,n_2,\,n_3) }{n_1\,n_2^{2-\alpha}} 
\simeq  
\frac{2}{1-\alpha}
\left[ \left(1 + \frac{h_2}{n_2}\right)^{1-\alpha}  - \frac{\alpha}{2-\alpha}  - \frac{\alpha(1-\alpha)}{(2-\alpha)(3-\alpha)}\left(1 + \frac{h_1}{n_1}\right)^{-(2-\alpha)} \right],
\qquad \alpha\neq 1,
\end{align}
\end{widetext}
for $\alpha<1$, this scales asymptotically as,
\begin{eqnarray}
\phi_\alpha\left(\frac{h_2}{n_2}\right)
%= \frac{\phi_\alpha(n_1,\,n_2,\,n_3) }{n_1\,n_2^{\gamma}}
\sim \left\{
     \begin{array}{lll}
       \frac{2}{1-\alpha} \left( \frac{h_2}{n_2} \right)^{1-\alpha}, &\;\text{for }
       & h_2 \gg n_2,
       \\[2mm] 
       \frac{6}{3-\alpha}\,, &\;\text{for }      
       & h_2 \ll n_2,
       \; h_1 \ll n_1,
       \\[2mm]
       \frac{4}{2-\alpha}\,,  &\;\text{for }     
       & h_2 \ll n_2,
       \; h_1 \gg n_1.
     \end{array}
     \right .\quad
     \label{eq:SM-3point-scaling}
\end{eqnarray}
\end{subequations}
We hence predict a data collapse for the $3$-point position auto-correlation when plotting the l.h.s of \Eq{SM-3point-scaling} as a function of $h_2/n_2$. 
For the regime when $h_2 \gg n_2$, one can  observe the power law as $1-\alpha$,
while for $h_2\ll n_2$, the correlation converges to some constants. Hence, \EQ{SM-3point-scaling} provides a new way of analysation for the position auto-correlation function, that depends only on a single parameter $h_2/n_2$, and the scaling form for different time-composition becomes irrelevant~\cite{GRTV17, VRTGM21}. 
The scaling of the $3$-point position auto-correlation function captures salient features which commonly observed in the anomalous transport dynamics, for instance when all times $n_1,\,n_2$ and $n_3$ are far separated and large enough, one commonly observes that correlation 
grows
%decays 
with $n_1 n_2^{2-\alpha}$, like $1/(n_1 n_2^{2-\alpha})$, 
in accordance with the prediction of \Eq{SM-3point-scaling}.
In section~\ref{subsec:sclngtest} we investigate how far these qualitative findings are substantial for quantitative comparison to the LLG~\cite{BF99,BCV10} that do not have mathematical findings on the position auto-correlation function.
%%%%%%%%%%%%%%%%%%%%%%%%%%%%%%%%%%%%%%%%%%%%%%%%%%%%%%%%%%%%%%%%%%%%%%%%%%%%%%%%%%%%%%%%%%%%%%%%%%

%\begin{oss}\label{rem:mom-corr-SM}
%For $0<\alpha<3$ and all $n$'s tends to infinity,
%%$n\rightarrow\infty$ 
%the $3$-point position auto-correlation $\phi_\alpha$ has the same asymptotic scaling, \Eq{3ptCORR-scal}, as for the third moment of displacement, \Eq{3-MOMNT-scaln}, i.e.,
%\begin{equation*}
%\phi_\alpha(n_1,n_2,n_3) \sim \langle (x_n-x_0)^3\rangle.
%\end{equation*}
%\end{oss}
%%%%%%%%%%%%%%%%%%%%%%%%%%%%%%%%%%%%%%%%%%%%%%%%%%%%%%%%%%%%%%%%%%%%%%%%%%%%%%%%%%%%%%%%%%%%%%%%%%
\subsubsection{Moments of velocity}
\label{sec:Moments-v-SM} 
The velocity of any point of the SM is either 
$+1$~or~$-1$ and moments of the velocity can be determined by evaluating
\begin{eqnarray}\label{eq:V-moment}
\langle v^p(n)\rangle = 2\sum\limits_{k=1}^{n}v_k^p(n)\Delta_k(\alpha) + 2\sum\limits_{k=n+1}^{\infty}v_k^p(n)\Delta_k(\alpha)\, ,
\end{eqnarray}
where $v_k(n)$ is the velocity at time $n$ of particle with $x \in [\ell^+_{k-1}, \ell^+_{k})$ where $\ell^+_{k}=1-\ell_{k}$ and $\Delta_k(\alpha)= \ell^+_{k} - \ell^+_{k-1} =\alpha/(k^{\alpha+1})(1+o(1))$. The velocity of the particle is given by
\begin{equation}\label{eq:velocity}
v_k(n)= \mathcal{I}_{\{n< k\}}-(-1)^{n-k}\mathcal{I}_{\{n \ge k\}} ,
\end{equation}
where $\mathcal{I}_A$ is the indicator of the event $A$. Then by using \Eq{velocity} in \eq{V-moment}, the moments of velocity switch between even and odd values of $p$. 
The even moments of velocity scale asymptotically like 
\begin{subequations}\label{eq:SM-vlcty-asymt-bhvr}
\begin{equation}
\langle v^p(n)\rangle \sim 1,\quad \text{as} \quad  n\rightarrow\infty,  \quad \text{even}\; p\geq 2.  
\end{equation}
The odd $p\geq 1$ moments of velocity  scales  asymptotically as
\begin{equation}\label{eq:asympt_behv}
\langle v^p(n)\rangle \sim  \left\{
     \begin{array}{rlll}
       1 - 4 \,R_\alpha\, ,       
       & \text{ for }
       & \text{even} 
       & n,
       \\[2mm]
       -1 + 4 \,R_\alpha\, ,       
       & \text{ for }
       & \text{odd} 
       & n,
     \end{array}
     \right .
\end{equation}
\end{subequations}
as $n$ changes between even and odd values, where
\begin{equation*}
R_\alpha = \sum\limits_{k=1}^{\infty}\Delta_{2k}(\alpha).
\end{equation*}
%This behaviour is not shared by the moments of velocity of the FnD dynamics \cf~\Eq{MoV-scaling-SM}. Therefore these observable can be used to distinguish the FnD dynamics and the SM.
%%%%%%%%%%%%%%%%%%%%%%%%%%%%%%%%%%%%%%%%%%%%%%%%%
\subsubsection{Velocity auto-correlation function}
\label{sec:Corr-v-SM}
The velocity of any point of the SM is either $+1$ or $-1$, and its auto-correlation is defined by
\begin{align}
&\langle v(n_1)v(n_2)\rangle =\nonumber \\
&\quad 2\sum\limits_{k=1}^{n} v(n_1)v_k(n_2)\Delta_k(\alpha) + 2\sum\limits_{k=n+1}^{\infty} v(n_1)v_k(n_2)\Delta_k(\alpha),
\end{align}
where $v_k(l)$, the velocity at time $l$ of a particle with position
$x\in[\ell^+_{k-1},\ell^+_k)$, is given in \Eq{velocity}. For $n_1=0$, we have $v(0)=1$, hence
\begin{equation}
\langle v(0)v(n_2)\rangle =
 2\sum\limits_{k=1}^{n} v_k(n_2)\Delta_k(\alpha) + 2\sum\limits_{k=n_2+1}^{\infty} v_k(n_2)\Delta_k(\alpha).
\end{equation}
%Likewise the odd moments of velocity, 
Calculations analogous to the previous ones, now show that the velocity auto-correlation function oscillates asymptotically in $n_2$ between two values. 
Therefore velocity auto-correlation follows the same asymptotic scaling, \Eq{asympt_behv}, as in the odd moments of velocity
\begin{equation}
\langle v(0)v(n_2)\rangle \sim \langle v^p(n) \rangle,  \quad \text{as} \quad  n\rightarrow\infty,  \quad \text{odd}\; p\geq 1.
\end{equation}
The $2$-times velocity auto-correlation function  is also asymptotically split into two cases
\begin{itemize}
	\item
	when $n_1$ and $n_2$ are either both even or both odd, then
\begin{equation}\label{eq:SM-VACF_EE-OO}
\langle v(n_1)v(n_2)\rangle \rightarrow 1, \qquad \text{as} \quad n_1\rightarrow \infty,\quad n_2>n_1,
\end{equation}
	\item
	when one of the two times is even and the other is odd, then
\begin{equation}\label{eq:SM-VACF_EO-OE}
\langle v(n_1)v(n_2)\rangle \rightarrow -1, \qquad \text{as} \quad n_1\rightarrow \infty,\quad n_2>n_1\,.
\end{equation}
\end{itemize}
%%%%%%%%%%%%%%%%%%%%%%%%%%%%%%%%%%%%%%%%%%%%%%%%%
\subsection{The fly-and-die dynamics}
\label{sec:FnD-model}
In the FND dynamics, we label trajectories by their initial position, $x_0$.  Until time
$t_c(x_0)$ such a trajectory moves along the positive $x$ axis with
unit velocity.  At time $t_c(x_0)$ it stops and remains at position
$x_0 + t_c(x_0)$ for all later times.  Accordingly, we call this FND dynamics.  Its position at time $t$ will be denoted
as
\begin{subequations}
\label{eq:FnD-EOM}
\begin{equation}
  x( x_0, t ) = \left\{
  \begin{array}{lcl}
    x_0 + t  \, ,        & \text{for} & t \leq t_c(x_0)\,,
    \\[2mm]
    x_0 + t_c( x_0 )\, , & \text{for} & t \geq  t_c(x_0)\,.
  \end{array}
  \right .
\end{equation}
Superdiffusive motion is expected to emerge when the distribution of
the times for the flights, $t_c(x_0)$ has a power-law tail.  To be
concrete, we consider here the case
\begin{equation}\label{eq:FnD_tail}
  t_c(x_0) = \left( \frac{l}{ x_0 } \right)^{1/\mu},
\end{equation}
\end{subequations}
with initial conditions, $x_0$ uniformly distributed in the interval $[0, 1]$, and $\mu >0$.  In the following, we explore the position and velocity moments and correlations of this ensemble of trajectories. The ensemble average is denoted by $\langle \cdot \rangle$.
The probability $P(>t)$ to perform a flight longer than $t$ 
amounts to the fraction of initial condition $x_0$ with $t_c(x_0) > t$
such that
\begin{equation}\label{eq:FnD_Probability}
  P(>t) = x_0(t) = \frac{ l }{ t^{\mu}  } \, .
\end{equation}
%%%%%%%%%%%%%%%%%%%%%%%%%%%%%%%%%%%%%%%%%%%%%%%%%%
\subsubsection{$p^{th}$ position moments }
\begin{lem}\label{lemma:FnD-momnts}
For $\mu>0$, the p$^{th}$ position moment of the FND for the trajectories starting at initial position $x_0$, asymptotically scales as
\begin{equation}
  \langle | \Delta x(t) |^p \rangle
  \sim 
  \left\{
  \begin{array}{lll}
    \frac{\mu}{\mu - p} \: l^{p/\mu}\,,
    & \text{ for }
    & p < \mu \, ,
    \\[2mm]
    l \: \ln\frac{t^\mu}{l} \,,
    & \text{ for }
    & p = \mu \, ,
    \\[2mm]
    \frac{p \, l}{p - \xi}  \: t^{p - \mu} \,,
    & \text{ for }
    & p > \mu \, .
  \end{array}
  \right .
\label{eq:FnD-moment-scaling}
\end{equation}
\end{lem}
\begin{proof}
See appendix \ref{app:pth-momnt-FnD}\,.
\end{proof}

More in detail, for a specific case $p=2$, the MSD scales $\langle | \Delta x(t) |^2 \rangle
  \sim  t^\gamma$, where $\gamma = 2 - \mu$ and $0<\gamma<2$, this exhibits the wide spectrum of diffusion, when the transport exponent $\gamma <1$, this yields to sub-diffusion; $\gamma = 1$ this grows linearly in time (\ie normal diffusion), and for $\gamma > 1$ it is super-diffusive. Thus the FND dynamics capture all the transport regimes computed for SM. The FND dynamics for $p>\mu$, when adopting $l\equiv 2$ and $\mu=\alpha$, captures all the position moments which are computed for SM~\Eq{SM-moments}.
%%%%%%%%%%%%%%%%%%%%%%%%%%%%%%%%%%%%%%%%%%%%%%%%%
%%%%%%%%%%%%%%%%%%%%%%%%%%%%%%%%%%%%%%%%%%%%%%%%%
\subsubsection{Generalized position auto-correlation function}
The generalized (or $n$-point) position autocorrelation function for the FND dynamics is defined as
\begin{align}
  &\rho_\mu(t_1, t_2, \cdots t_n)\,, \\
  &=\langle \Delta x(t_1) \; \Delta x(t_2) \cdots \Delta x(t_n)\rangle  \,,
\nonumber\\[2mm]
  &=
  \langle \left( x( x_0, t_1 ) - x_0 \right) \: \left( x( x_0, t_2 ) - x_0 \right) \cdots \: \left( x( x_0, t_n ) - x_0 \right)\rangle \,,
\nonumber\\[2mm]
  &=
  \int_0^1 \diff x_0 \: \left( x( x_0, t_1 ) - x_0 \right) \left( x( x_0, t_2 ) - x_0 \right) \cdots \left( x( x_0, t_n ) - x_0 \right),
\nonumber
\end{align}
where it is assumed that $t_1<t_2<\cdots < t_n$. To evaluate the integral we follow the convention that $t_n$ is always larger or equal to $t_1$.
Accordingly, we split the integration range into $n$ intervals
\begin{description}
\item[$0 < x_0 < P(>t_n)$ ] 
	The trajectories are still flying at time $t_n$ such that
  $\Delta x(t_1) = t_1, \:\Delta x(t_2) = t_2, \cdots, \Delta x(t_{n-1}) = t_{n-1}, \:\Delta x(t_n) = t_n$.
  %or $\Delta x(t_n) =  \prod\limits_{k=1}^n t_k$.
\item[$P(>t_n) < x_0 < P(>t_{n-1})$ ] 
	The trajectories are still flying until time $t_{n-1}$ but it has died by the time $t_{n}$.
  Consequently, $\Delta x(t_1) = t_1,\:\Delta x(t_2) = t_2,\: \cdots, \Delta x(t_{n-1}) = t_{n-1}$ and $\Delta x(t_n) = t_c(x_0)$.\\
 \vdots
\item[$P(>t_{1}) < x_0 < 1$ ] 
	The trajectories died before $t_1$. Consequently, $\Delta x(t_{1}) = \Delta x(t_2) = \cdots = \Delta x(t_n) = t_c(x_0)$.
\end{description}
Splitting the integral and performing a calculation allows us to interpret it as follows
%\begin{align}
%&\phi(t_1,t_2,\cdots t_n) \\
%&= \int\limits_0^{P(>t_n)} \diff x_0 \: \left( x( x_0, t_1 ) - x_0 \right) \: \left( x( x_0, t_2 ) - x_0 \right) \cdots \: \left( x( x_0, t_n ) - x_0 \right)\nonumber\\  
%&+
%\int\limits_{P(>t_n)}^{P(>t_{n-1})} \diff x_0 \: \left( x( x_0, t_1 ) - x_0 \right) \: \left( x( x_0, t_2 ) - x_0 \right) \cdots \: \left( x( x_0, t_n ) - x_0 \right) \,  \nonumber\\
%&\vdots \nonumber \\
%&+
%\int\limits_{P(>t_{1})}^1 \diff x_0 \: \left( x( x_0, t_1 ) - x_0 \right) \: \left( x( x_0, t_2 ) - x_0 \right) \cdots \: \left( x( x_0, t_n ) - x_0 \right) \,.
%\end{align}
%Thus the formal derivation can be interpret as following
\begin{widetext}
\begin{align}
\rho_\mu(t_1,t_2,\cdots t_n) &= (t_1 t_2 \cdots t_n)\int\limits_0^{l/t_n^\mu} \diff x_0 
+ 
(t_1t_2 \cdots t_{n-1})\int\limits_{l/t_{n}^\mu}^{l/t_{n-1}^\mu} \diff x_0 \left(\frac{l}{x_0}\right)^{\frac{1}{\mu}}
+
(t_1t_2 \cdots t_{n-2})\int\limits_{l/t_{n-1}^\mu}^{l/t_{n-2}^\mu} \diff x_0 \left(\frac{l}{x_0}\right)^{\frac{2}{\mu}}
%(t_1\:t_2 \cdots t_{n-3})\int\limits_{l/t_{n-2}^\mu}^{l/t_{n-3}^\mu} \diff x_0 \left(\frac{l}{x_0}\right)^{\frac{3}{\mu}}
+
\cdots + 
\int\limits_{l/t_{1}^\mu}^1 \diff x_0 \left(\frac{l}{x_0}\right)^{\frac{n}{\mu}}, \nonumber\\
&= \sum\limits_{j=0}^{n} \left( \prod\limits_{k=1}^{n-j} t_k
\int\limits_{l/{t^\mu_{n-j+1}}}^{l/{t^\mu_{n-j}}} \diff x_0 \left(\frac{l}{x_0}\right)^{\frac{j}{\mu}} \right).\label{eq:M-point-FnD-summ}
\end{align}
\end{widetext}
Simple integration allows us to write a general expression of $n$-point position auto-correlation function as
\begin{eqnarray}\label{eq:npt-corr-FnD-solv}
\rho_\mu(t_1,t_2,\cdots t_n) = l \sum\limits_{j=0}^{n}  \prod\limits_{k=1}^{n-j}t_k
\left( \frac{\mu}{\mu-j} \left( t_{n-j}^{j-\mu}  - t_{n-j+1}^{j-\mu} \right)\right),\quad
\end{eqnarray}
where $t_{n+1} = \infty$ and $t_0 = l^{1/\mu}$.
When adopting $l\equiv 2$ and $\mu=\alpha$, correlation \Eq{npt-corr-FnD-solv}
yields the same scaling as find for the $m$-point position auto-correlation of the SM, \Eq{N-point-SM-summ}. Therefore the higher order position auto-correlation function of the SM and the FND asymptotically scales in the same trend.
%%%%%%%%%%%%%%%%%%%%%%%%%%%%%%%%%%%%%%%%%%%%%%%%%%%%%%%%%%%%%%%%%%%%%%%%%%%%%%%%%%%%%%%%%%%%%%%%%%%%
%In the later, we derive the $3$-point position correlation function. When adopting $n=3$ and performing calculations, \Eq{npt-corr-FnD-solv} yields the $3$-point position auto-correlation function. Therefore, for $\mu \neq 1$, we find
In the subsequent, we derive the $3$-point position correlation function. Upon setting $n=3$ and performing calculations on \Eq{npt-corr-FnD-solv}. Consequently, when $\mu \neq 1$, we find
\begin{equation}\label{eq:3p-FnD-Corr}
\rho_\mu(t_1,\,t_2,\,t_3) \simeq 
\frac{l\,t_1\,t_2\,t_3^{1-\mu}}{1-\mu} 
- \frac{l\,\mu\,t_1\,t_2^{2-\mu}}{(2-\mu)(1-\mu)} 
- \frac{l\,\mu\, t_1^{3-\mu}}{(2-\mu)(3-\mu)} \,.
%\nonumber \\[2mm]
%+ \mathcal{O}(1)
%&-\frac{l^{\frac{3}{\mu}}\,\mu}{3-\mu}\,.
\end{equation}
\begin{itemize}
    \item[I.]
      		 At a fixed time $t_1$, when considering two equivalent subsequent times $t_2$ and $t_3$, the \emph{MSD} (\cf~\Eq{FnD-moment-scaling} with $p=2$) exhibits an asymptotic scaling 
\begin{equation}\label{eq:2-MOMNT-scaln-FnD}
\langle | \Delta x(t) |^2 \rangle \sim \frac{2\,l}{2-\alpha}t^{2-\mu}\,, \quad 0<\mu< 2\,.
\end{equation}
	\item[II.]
			For $t_1=t_2=t_3$, this reduces to the third moment for the displacement, \Eq{FnD-moment-scaling} with $p=3$,
\begin{equation}\label{eq:3-MOMNT-scaln-FnD}
  \langle | \Delta x(t) |^3 \rangle =
  \frac{3\, l}{3-\mu} t^{3-\mu},\quad 0<\mu<3.
 % \: \left( t^{3-\mu} - \frac{2 - \zg}{3} \: b^{(\zg+1)/(2-\zg)} \right).
\end{equation}
\end{itemize}
%%%%%%%%%%%%%%%%%%%%%%%%%%%%%%%%%%%%%%%%%%%%%%%%%%%%%%%%%%%%%%%%%%%%%%%%%%%%%%%%%%%%%%%%%%%%%%%%%%
Some functional relationships between the times are defined as follows
\begin{enumerate}
\item[1.] 
varying $t_1$, $t_2$ and $t_3$ while keeping a fixed time lag, with $h_1=t_2-t_1$ and $h_2=t_3-t_2$.
%varying $t_1$, $t_2$ and $t_3$ for a fixed time lag $h_1=t_2-t_1$ and $h_2=t_3-t_2$; 
\item[2.] varying $t_1$ while setting $t_2=t_1+t_1^{q_1}$ and $t_3=t_1+t_1^{q_1}+t_1^{q_2}$ for some fixed value of $q_1<1$ and $q_2<1$.
\item[3.] fixing $t_1 \leq t_2$ fixed while setting $t_3$ to vary. 
\end{enumerate}
Like the SM, when all \( n \)-values approach infinity for the 3-point position auto-correlation function \(\rho_\mu(t_1, t_2, t_3)\), as defined in \Eq{3p-FnD-Corr}, the power law exponent exhibits the same scaling behavior as that found for the third position moment (see Lemma \ref{lemma:nptscaling-FnD} for the case where \( m = 3 \)). To better understand this, we need to look at how the correlation \(\rho_\mu(n_1, n_1 + h_1, n_1 + h_1 + h_2)\) behaves when \( n_1 \) is very large.
In this context, for large \( n_1 \), if the time lags \( h_{k-1} \) (for \( k = 2, 3 \)) are either constant or grow as \( n^q \) with \( q < 1 \), the differences between the three times \( n_1 \), \( n_1 + h_1 \), and \( n_1 + h_1 + h_2 \) become insignificant compared to the average value of \( n_1 \). This means that as \( n_1 \) increases, the relative differences between these times diminish, leading to a simpler scaling relationship for the correlation function.
%%%%%%%%%%%%%%%%%%%%%%%%%%%%%%%%%%%%%%%%%%%%%%%%%%
%%%%%%%%%%%%%%%%%%%%%%%%%%%%%%%%%%%%%%%%%%%%%%%%%%
\begin{lem}\label{lemma:nptscaling-FnD}
For $0<\alpha<n$, as all $t$ times tend to infinity and $h_{k-1} = t_k - t_{k-1}$, for $k=2,3, \cdots, n,$ where $h_{k-1}$ is either fixed or $\sim t_1^q$, with $q<1$, the $n$-point position auto-correlation function $\rho_\mu$ as represented in Equation \Eq{3p-FnD-Corr}, asymptotically scales as
\begin{equation}\label{eq:nptCORR-scal-FnD}
    \rho_\mu(t_1,t_2,\cdots, t_n) \sim \frac{n\,l}{n-\mu} t_1^{n-\mu}\,, \quad 0<\mu<n\,.
\end{equation}
\end{lem}
\begin{proof}
This follows directly from \Eq{3p-FnD-Corr}.\hfill$\square$
\end{proof}
%%%%%%%%%%%%%%%%%%%%%%%%%%%%%%%%%%%%%%%%%%%%%%%%
\begin{oss}\label{rem:mom-corr-SM}
For $0<\mu<n$, with $n=2,3$ when all times $t$'s tend to infinity, the $2$- and $3$-point position auto-correlation function $\rho_\mu$,
following the $h_{k-1}$ represented in Lemma~\ref{lemma:nptscaling-FnD}, have the same asymptotic scaling, \Eq{nptCORR-scal-FnD}, as for the second i.e., MSD, \Eq{2-MOMNT-scaln-FnD} and third moments of displacement, \Eq{3-MOMNT-scaln-FnD}, i.e.,
\begin{equation}\label{eq:sam-sclng-momnt-corr-FnD}
\rho_\mu(t_1,t_2) \sim \langle | \Delta x(t) |^2 \rangle, \;\; \text{and} \;\;
\rho_\mu(t_1,t_2,t_3) \sim \langle | \Delta x(t) |^3 \rangle,
\end{equation}
respectively.
\end{oss}
Therefore, it is observed that when time is significantly large, the position moments and correlation function of the SM and the FND scale in the same way. Likewise the SM, a single parameter dependent scaling of the correlation function, reconsider \Eq{3p-FnD-Corr} and perform calculations, introducing the time lag $h_2=t_3-t_2$ and $h_1=t_2-t_1$, and normalizing by the factor $t_1t_2^{2-\mu}$, the position auto-correlation as a function of $h_2/t_2$ entails as
%%%%%%%%%%%%%%%%%%%%%%%%%%%%%%%%%%%%%%%%%%%%%%%%%%%%%%%%%%%%%%%%%%%%%%%%%%%%%%%%%%%%%%%%%%%%%%%%%%%%
\begin{subequations}\label{eq:U_3pt_corr_FnD}
\begin{widetext}
\begin{align}
\rho_\mu\left(\frac{h_2}{t_2}\right) = 
\frac{\phi_\mu(t_1,t_2,t_3)}{t_1t_2^{2-\mu}} \simeq \frac{l}{1-\mu} 
\left[ 
\left(1+\frac{h_2}{t_2}\right)^{1-\mu} 
-\frac{\mu}{2-\mu}
-\frac{\mu(1-\mu)}{(2-\mu)(3-\mu)} \left(1 + \frac{h_1}{t_1} \right)^{-(2-\mu)}\right],\qquad  \mu\neq 1. 
\end{align}
\end{widetext}
%%%%%%%%%%%%%%%%%%%%%%%%%%%%%%%%%%%%%%%%%%%%%%%%%%%%%%%%%%%%%%%%%%%%%%%%%%%%%%%%%%%%%%%%%%%%%%%%%%%%
In the large time limit the asymptotic scaling for large and small values of $h_2/t_2$ and $h_1/t_1$ for $\mu<1$ yields as
\begin{align}
\rho_\mu\left(\frac{h_2}{t_2}\right) \simeq \left\{
     \begin{array}{lll}
       \frac{l}{1-\mu} \left( \frac{h_2}{t_2} \right)^{1-\mu}, &\;\text{for  } 
       & h_2 \gg t_2,
       \\[2mm] 
       \frac{3l}{3-\mu}\,, &\;\text{for  }       
       & h_2 \ll t_2, \;
        h_1 \ll t_1, 
       \\[2mm]
       \frac{2l}{2-\mu}\,, & \;\text{for  }             
       & h_2 \ll t_2,\;
        h_1 \gg t_1.
     \end{array}
     \right .
     \label{eq:FnD-3point-scaling}
\end{align}
\end{subequations}
%    \begin{array}{lll}
%      t^{\frac{p}{1+\xi}} \,,                 & \quad\text{for  } & \xi<1,\ p<\xi    \, ,\\
%      t^{\frac{p(1+\xi)-\xi^{2}}{1+\xi}} \,,     & \quad\text{for  } & \xi<1,\ p>\xi    \, ,\\
%      t^{\frac{p}{2}} \,,                    & \quad\text{for  } & \xi>1,\ p<2\xi-1 \, ,\\
%      t^{\frac{1}{2}+p-\xi} \,,                & \quad\text{for  } & \xi>1,\ p>2\xi-1 \, . 
%    \end{array}
%  \right. 

This scaling is identical to the SM expression when $l\equiv 2$, \Eq{SM-3point-scaling} for large and small times. Therefore the asymptotic scaling of the 3-point position correlation as a function of $h_2/t_2$, the SM and the FND scales is in a similar fashion (\cf~\Eqs{SM-3point-scaling} and  \eq{FnD-3point-scaling}). Hence we can predict data collapse of $3$-point position correlation irrespective of the times relationship. In section \ref{subsec:sclngtest}, we emphasize this fact by comparing the qualitative prediction with LLG~\cite{BCV10,BF99,BFK00}.

%We start by recalling the scaling form of $m$-point correlation of SM represented in \Eq{N-point-SM-summ}. 

%see how far it captures the 3-point position correlation3 scaling of the SM. In this regard, we recall Eq. (A.19), perform simple computation and take ? = 2 ? ?, that allow us to write 3-point position correlation ?(o, n,m) scaling of the SM in the following form

%%%%%%%%%%%%%%%%%%%%%%%%%%%%%%%%%%%%%%%%%%%
\subsubsection{Moments of velocity}
In the FND dynamics, the velocity of each trajectory is $+1$, these flying trajectories contribute to the velocity moments where other trajectories stop $v=0$ and do not contribute.
Therefore only those trajectories will contribute those are still flying $v=1$. 
The moments of the velocity $\langle v^p(t) \rangle$ obtain as
\begin{align*}
\langle v^p(t) \rangle &= \langle | v(x_o,t) - v_0 |^p \rangle \,, \\ 
&= \int\limits_0^1 dx_0 \,|v(x_o,t) - v_0 |^p 
= \int\limits_0^{P(>t)} dx_0 \, t^p 
= \int\limits_0^{l/t^\mu} dx_0 \,,
\end{align*}
which asymptotically scales as
\begin{align}\label{eq:FnD-velocity-momnts}
\langle v^p(t) \rangle \sim l\,t^{-\mu}, \qquad p\geq \mu\,.
\end{align}
This behaviour is not shared by the velocity moments of the SM, \EQ{SM-vlcty-asymt-bhvr}. 

%%%%%%%%%%%%%%%%%%%%%%%%%%%%%%%%%%%%%%%%%%%%%%%%
%%%%%%%%%%%%%%%%%%%%%%%%%%%%%%%%%%%%%%%%%%%
\subsubsection{Velocity auto-correlation function}
The velocity of each flying trajectory in the FND dynamics is $+1$.
The trajectories are flying with the velocity $v=1$, till they stop, $v=0$. Therefore only those trajectories contribute to the velocity auto-correlation functions with $t_1\leq t_2 \leq \cdots t_n$, that are still flying at time $t_n$. Thus denoting velocity correlation $\rho_v(t_1,t_2,\cdots t_n)$, we obtain
\begin{align*}
&\rho_v(t_1,t_2,\cdots t_n) \\
&= \langle \Delta v(t_1)\Delta v(t_2)\cdots\Delta v(t_n) \rangle  \,, \\
&=\langle \left( v( x_0, t_1 ) - v_0 \right) \: \left( v( x_0, t_2 ) - v_0 \right)\cdots \left( v( x_0, t_n ) - v_0 \right) \rangle \,,
\\
&=\int\limits_0^{1} \diff \,x_0\;\left( v( x_0, t_1) - v_0 \right) \: \left( v( x_0, t_2 ) - v_0 \right)\cdots \left( v( x_0, t_n ) - v_0 \right)\,,  
\\
&= \int\limits_0^{P(>t_n)} \diff \,x_0\;\left( v( x_0, t_1) - v_0 \right) \: \left( v( x_0, t_2 ) - v_0 \right)\cdots \left( v( x_0, t_n ) - v_0 \right) \,, \\
&= \int\limits_0^{l/t_n^\mu} \diff x_0\, , 
\end{align*}
therefore velocity auto-correlation asymptotically scales as
\begin{align}
\rho_v(t_1,t_2,\cdots t_n) \simeq l\,t_n^{-\mu}\,, \qquad n > \mu \,,
\end{align}
where $t_n = t_{n-1} + h_{n-1}$, $n \in \{2,3,\cdots\}$ and $h>0$.

This exhibits the same power law tail, $-\mu$ for any order of velocity correlation function, moreover velocity moments and correlation asymptotically scale in same power law behavior (\cf~\Eq{FnD-velocity-momnts}). 
For $n=2$, we find $2$-point velocity correlation function $\rho_v(t_1,t_2) \simeq l\,t_1^{-\mu}\,, \;0<\mu<2$. This behaviour is not shared by the $2$-time velocity auto-correlation function of the SM (\cf~\Eqs{SM-VACF_EE-OO} or \eq{SM-VACF_EO-OE}), thus can be used to distinguish the transport processes.

\begin{figure*}[t]%{figure}
  \centering
    \includegraphics[width=8.5cm, height=6cm]{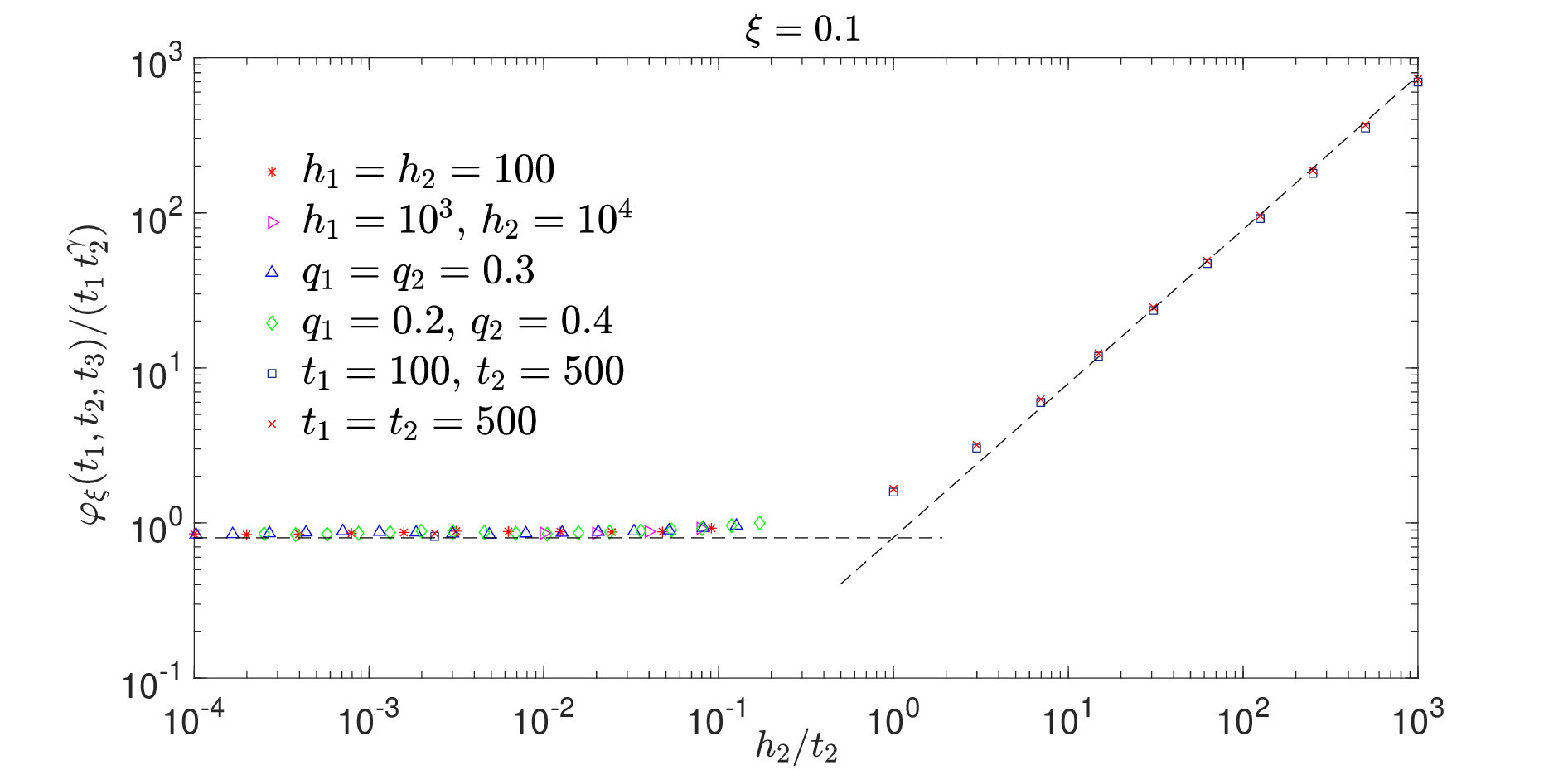}
    \includegraphics[width=8.5cm, height=6cm]{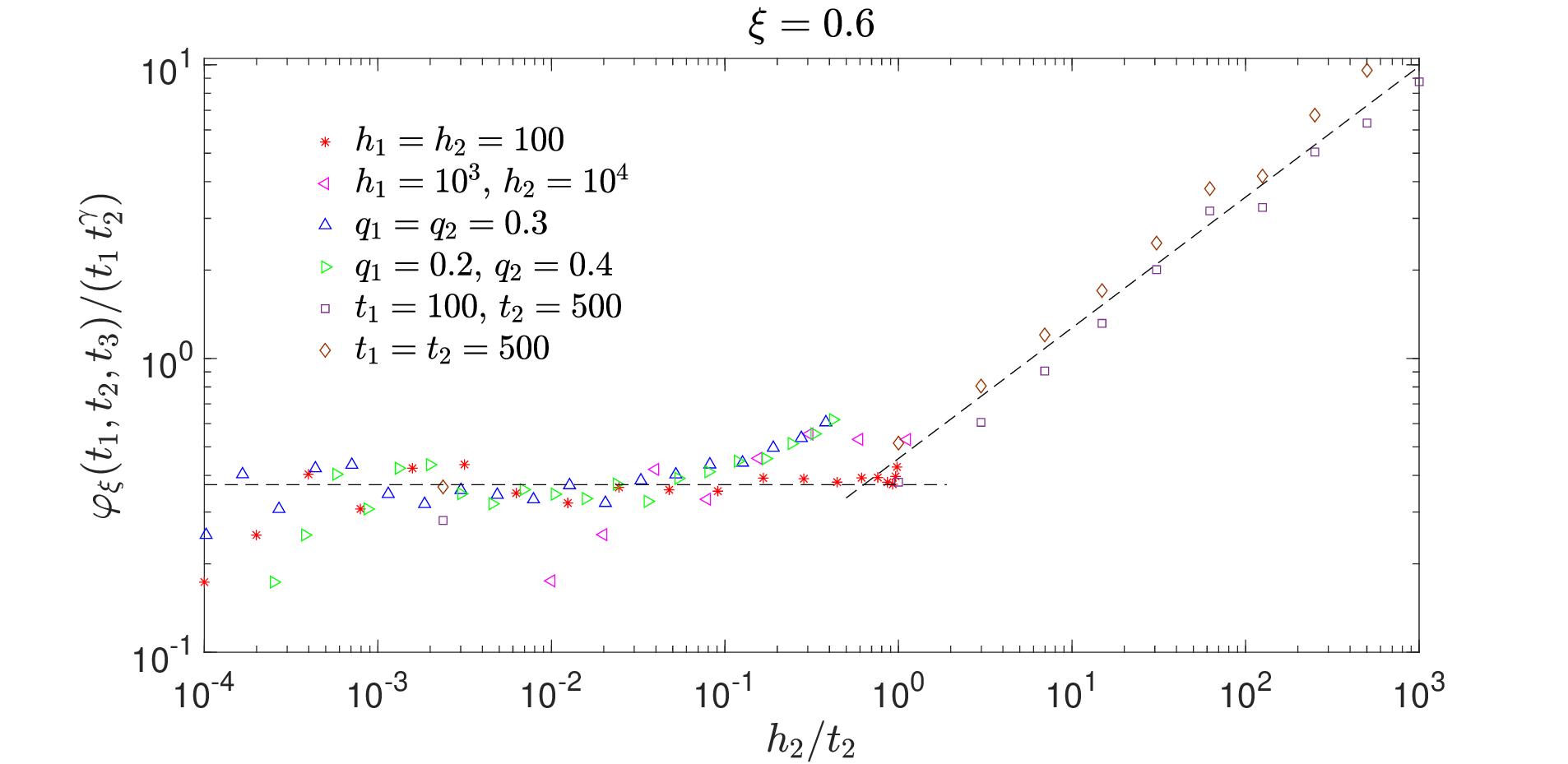}
     \caption{\label{fig:3-point-data}The $3$-point position auto-correlation functions $\varphi_\xi(t_1, t_2,t_3)$ of the LLG are plotted for $\xi=0.1$ (left panel) and $\xi=0.6$ (right panel). We obtain a data collapse for a vast data set of combinations of $t_1,\,t_2$ and $t_3$ by plotting the left-hand side of \Eqs{U_3pt_corr} or \eq{U_3pt_corr_FnD} as a function of $h_2/t_2$. The different symbols denote data for $d_0 = 0.1$, 
%where we varied $t_2$ and $t_3$ at fixed time lag $h_1 = t_2 - t_1$ and $h_2 = t_3 - t_2$, 
%time $t_1$, 
where we varied the time $t_1$, while setting $t_2=t_1+h_1$ and $t_3=t_1+h_1+h_2$, for $h_1$ and $h_2$ any positive constants (\cf~legend) for the regime $h_2<t_2$, and also we varied $t_1$ while setting $t_2 = t_1 + t_1^{q_1}$ and $t_3 = t_1 + t_1^{q_1} + t_1^{q_2}$ (\cf~legend) for the regime $h_2<t_2$. Similarly, for $h_2>t_2$, keep $t_2\geq t_1$ fixed and vary $t_3$ (\cf~legend). The dashed lines show the parameter dependence \Eq{U_3pt_corr} predicted by the SM (or so FND, \Eq{U_3pt_corr_FnD}). 
     %The inset shows the mean $(\bar{X})$ and relative deviation $(\sigma)$ of the numerical data and the theoretical prediction (\ie the difference divided by the predicted value) for the regime $ \tau_2<t_2$. We observe that the data collapse gradually gets worse as $\xi$ gets larger, presumably due to poor statistics of numerical data.
     }
    \end{figure*}
%%%%%%%%%%%%%%%%%%%%%%%%%%%%%%%%%%%%%%%%%%%%%%%%%%%%%%%%%%%%%%%%%%%%%%%%%%%%%%%%%%%%%%%%%%%%%%%%%%
\section{Stochastic Process}
\subsection{The L\'evy-Lorentz gas\label{sec:LLg}}
The LLG was introduced in Barkai and Fleurov~\cite{BF99} as a one-dimensional model for anomalous
transport in semiconductor devices where diffusion
arises from scattering at dislocations at fixed random positions.
Subsequently, it has been investigated by many
authors \cite{BFK00,BCV10}. The LLG is a one-dimensional model
that comprises ballistic flights between scatterers at fixed
random positions. The distances $d$ between neighboring
scatterers are independently and identically distributed
random variables sampled from a L\'evy distribution with
probability density
%LLg is a one dimensional random walk in random environment, where scatterers are placed randomly and distance $d$ between two consecutivescatteres are independently and identically distributed random variable from L\'evy distribution then the probability density is given by
\begin{equation}\label{eq:probdnsty}
\lambda(d) \equiv  \frac{\xi}{d_0} \left(\frac{d}{d_0}\right)^{-\left(\xi+1\right)}, \qquad  d\in [d_0,\infty),
\end{equation}
where $\xi>0$ and $d_0$ is the minimum distance between
scatterers. A point particle moves ballistically with velocity $\pm v$ between the two consecutive scatterers when it hits a scatterer, then it is either transmitted or reflected by the probability $1/2$.
%where \xi > 0 and $r_0$ is the minimum distance between scatterers. In this system a particle moves ballisticaly with velocity ±v between any two consecutive scatterers, and when it hits a scatterer it is transmitted or reflected with probability 1/2.
Barkai et al.~\cite{BFK00} calculated bounds for the MSD for equilibrium and non-equilibrium initial conditions.
Subsequently, Burioni et al.~\cite{BCV10} adopted some simplifying assumptions to find the asymptotic form 
for non-equilibrium conditions of all moments $\langle |d(t)|^{p}\rangle$ with $p>0$\,
\begin{eqnarray}\label{eq:MOMBUR}
  \langle |d(t)|^{p} \rangle
  \sim
  \left\{
    \begin{array}{lll}
      t^{\frac{p}{1+\xi}} \,,                 & \quad\text{for  } & \xi<1,\ p<\xi    \, ,\\
      t^{\frac{p(1+\xi)-\xi^{2}}{1+\xi}} \,,     & \quad\text{for  } & \xi<1,\ p>\xi    \, ,\\
      t^{\frac{p}{2}} \,,                    & \quad\text{for  } & \xi>1,\ p<2\xi-1 \, ,\\
      t^{\frac{1}{2}+p-\xi} \,,                & \quad\text{for  } & \xi>1,\ p>2\xi-1 \, . 
    \end{array}
  \right. 
\end{eqnarray}
For the MSD, $p=2$, this result implies
\begin{eqnarray}\label{eq:LLgMeanSquare}
  \langle d(t)^{2} \rangle 
  \sim 
  t^\eta, \;\;\;\;
  %\quad\text{with}\quad 
  \eta 
  =
%  \left\{
%    \begin{array}{llr@{\;\xi\;}l}
%      2 - \frac{\xi^{2}}{(1+\xi)}   & \quad\text{for  } & &  < 1    \, ,\\
%      \frac{5}{2}-\xi              & \quad\text{for  } & 1   \leq &  < 3/2  \, ,\\
%      1                            & \quad\text{for  } & 3/2 \leq &         \, .
%    \end{array}
%  \right. 
\left\{
    \begin{array}{lll}
      2 - \frac{\xi^{2}}{(1+\xi)} \,,   & \quad\text{for  } & \xi < 1    \, ,\\
      \frac{5}{2}-\xi \,,             & \quad\text{for  } & 1  \leq \xi < 3/2  \, ,\\
      1  \,,                          & \quad\text{for  } & 3/2 \leq \xi         \, .
    \end{array}
  \right. 
\end{eqnarray}%
Unlike the SM and the FND, which enjoys sub-diffusive transport for $\alpha > 1$ and $\mu > 1$ respectively, non-equilibrium initial conditions for the LLG only lead to 
super-diffusive $(0<\xi< 3/2)$ or 
diffusive $(\xi \ge  3/2)$ regimes: sub-diffusion is not expected.

%Salari et al.~\cite{Salari} observed that 
The moments of the SM in its super-diffusive regime ($0<\alpha < 1$) can be mapped to those of the LLG~\cite{Salari,GRTV17}. 
All moments of the SM, \Eq{SM-moments}, (and so FND, \Eq{FnD-moment-scaling}) scale like those conjectured and numerically validated for the LLG, \Eq{MOMBUR}, once the second moments do. 
This is the case if the following holds (\cf~\Eqs{SM-moments} and \eq{FnD-moment-scaling})
%\cf~\Eqs{MeanSquareDisplacement} and \eq{LLgMeanSquare}:
\begin{equation}\label{eq:llgpdf}
  \alpha = \mu
  =
  \left\{
    \begin{array}{lll}
      \frac{\xi^{2}}{(1+\xi)}\,,  & \quad\text{for  } & 0 < \xi \leq 1          \, ,\\[2mm]
      \xi - \frac{1}{2}\,,    & \quad\text{for  } & 1 < \xi \leq \frac{3}{2}\, ,\\[2mm]
      1  \,,       & \quad\text{for }  & \frac{3}{2} < \xi \, .
    \end{array}
  \right. 
\end{equation}
When adopting this mapping all other moments of the SM, the FND and the LLG agree with those of the LLG, \Eq{MOMBUR}. 
This means that \Eq{llgpdf} make these processes asymptotically indistinguishable from the point of view of all position moments and $2$-point position auto-correlation function \cite{GRTV17,VRTGM21}. We thus now extend these equivalence to the $3$-point position auto-correlation function and check whether the higher correlations differ or they follow the same equivalence agreement. The single dimensionless time ratio $h_2/t_2$ expression for the $3$-point correlations are calculated analytically for the SM and the FND, (\cf~\Eqs{U_3pt_corr} and \eq{U_3pt_corr_FnD}). This will then be compared to numerically estimated data for the LLG. For the position auto-correlation function in the LLG, there are no analytic results of any order, such as those of Burioni et al.~\cite{BCV10} for the moments. We numerically estimate the $3$-point displacement correlation in~\ref{subsec:LLg_correlations}.
%%%%%%%%%%%%%%%%%%%%%%%%%%%%%%%%%%%%%%%%%%%%%%%%%%
%%%%%%%%%%%%%%%%%%%%%%%%%%%%%%%%%%%%%%%%%%%%%%%%
\begin{oss}
For $t\rightarrow \infty$, the asymptotic behavior of the third moment of displacement of the LLG can be obtained by requesting $p=3$ in \Eq{MOMBUR}, one finds
\begin{align}
\langle d(t)^3 \rangle \sim t^{\overline{\eta}}, \qquad \overline{\eta} = \eta + 1\,.
\end{align}
\end{oss}
\subsection{Generalized position auto-correlation function}
We define the generalized (or $n$-point) position auto-correlation function of the LLG as following
\begin{equation}\label{eq:npCORR-LLg}
\varphi_\xi (t_1,\,t_2,\,\cdots,\, t_n) = \mathbb{E}\,[ d(t_1)\,d(t_2)\,\cdots \,d(t_n)]\,,
\end{equation}
where $\mathbb{E}$ denotes the averages, first average over the particles and then on the given random scatterers realization. We intend to compare the asymptotic form of the position auto-correlation function with the SM and the FND.
\begin{lem}
For $\xi>0$, and all $t's$ are tends to infinity, the $n$-point position auto-correlation function $\varphi_\xi$, \Eq{npCORR-LLg}, of the LLG has the following asymptotic form
\begin{equation}\label{eq:np-corr-momnt}
\varphi_\xi(t_1,t_2,\cdots,t_n) \sim 
c(h_1,h_2\cdots\,  h_{n-1})
\;t_1^{\omega_p}\,,
% \;t_1^{\omega_p(h_1,h_2\cdots h_{n-1})}\,,
\end{equation}
for $p=2,\cdots,n$, where $h's$ represents the time difference, $c(h_1,h_2\cdots\,  h_{n-1})
$ denotes pre-factor and $\omega_p$ is the power law exponent for the respective order of the correlation function, this will be obtained by best fit to the data.
\end{lem}
We intend to compare the asymptotic form of the $3$-point position auto-correlation function with the SM and the FND. In what follows, we define the $3$-point position correlation function of the LLG. 
%%%%%%%%%%%%%%%%%%%%%%%%%%%%%%%%%%%%%%%%%%%%%%%%%%%%%%%%%%%%%%%%%%%%%%%%%%%%%%%%%%%%%%%%%%%%%%%%%%

\subsubsection{$3$-point position auto-correlation}\label{subsec:LLg_correlations}
We define the $3$-point position auto-correlation function of the LLG by requesting $n=3$ in \Eq{npCORR-LLg} as following
\begin{equation}\label{eq:3pCORR-LLg}
\varphi_\xi (t_1,\,t_2,\,t_3) = \mathbb{E}\,[ d(t_1)\,d(t_2)\,d(t_3)]\,,
\end{equation}
where $\mathbb{E}$ denotes the averages, first average over the particles and then on the given random scatterers realization. 
%We intend to compare the asymptotic form of the position auto-correlation function with the SM and the FND. 
Since our aim is to check the asymptotic equivalence of the position auto-correlation function with the SM and the FND. Thus, we follow the same time composition as adopted in Sec.~\ref{subsec:3pcorr} 
\begin{enumerate}  
\item  $h_1 = t_2 - t_1$, \; as $h_1>0$ either finite or $h_1\sim t_1^{q_1},\;q_1<1$, and $t_1\rightarrow \infty$\,.
\item  $h_2 = t_3 - t_2$, \; as $h_2>0$ either finite or $h_2 \sim t_1^{q_2},\;q_2<1$, where $t_2 = t_1+h_1$, and $t_1 \rightarrow \infty$\,.
\item $t_1\geq t_2$ are fixed and set $t_3\rightarrow\infty$\,.
\end{enumerate}
%\begin{oss}
%For $\xi>0$, and $t_1\rightarrow\infty$, the $3$-point position auto-correlation function 
%$\varphi_\xi$, \Eq{3pCORR-LLg}, of the LLG has following asymptotic form
%\begin{equation}\label{eq:3p-corr-momnt}
%\varphi_\xi(t_1,t_1+\tau_1,t_1+\tau_1+\tau_2) \;\sim \;c(\tau_1,\tau_2) \;t_1^{\omega_c(\tau_1,\tau_2)}\,,
%\end{equation}
%where $c(\tau_1,\tau_2)$ denotes pre-factor and $\omega_c(\tau_1,\tau_2)$ is the power law exponent, this will obtain by best fit to the data.
%\end{oss}
Asymptotic scaling form for the moments and the $2$-point position correlation of the SM and the LLG have been tested~\cite{Salari,GRTV17}, when $\alpha$ and $\xi$ obey \Eq{llgpdf}.
In the following, we now verify the theoretical prediction of the $3$-point position auto-correlation as a function of $h_2/t_2$ of the SM \Eq{U_3pt_corr} and FND \Eq{U_3pt_corr_FnD}, with numerically estimated correlations of the LLG, \Eq{3pCORR-LLg}, that $\alpha,\,\mu$ and $\xi$ 
%(where $\alpha=\mu=2-\gamma$) 
obey the same relation \Eq{llgpdf}. The importance of single qualitative scaling predicts the data collapse of the LLG for small and large $h_2/t_2$. 
%%%%%%%%%%%%%%%%%%%%%%%%%%%%%%%%%%%%%%%%%%%%%%%%%
\begin{oss}
The asymptotic behaviour of the $1$-time velocity auto-correlation function of the LLG scales like $\langle v(0)\:v(t)\rangle \sim t^{-3/2}$, as obtained by Barkai et al.~\cite{BF99}, hence it can be used to distinguish the LLG from the SM and the FND dynamics.
\end{oss}

\subsection{Scaling test of the $3$-point position auto-correlation function of the SM, FND, and LLG\label{subsec:sclngtest}}
%The analytical expression of the position auto-correlation function of the LLg still an hard finding. Although \cite{GRTV17} numerically estimated the $2$-point correlation function.  
In this section, we explore the equivalence of the $3$-point position auto-correlation of the SM, the FND, and the LLG. We try here to extend this equivalence to the $3$-point correlations.
Since the $2$-point position correlation provides the faithful description of these systems~\cite{GRTV17}. We start by recalling the scaling of $3$-point position correlation represented in \Eqs{U_3pt_corr} or \eq{U_3pt_corr_FnD} and see how far it captures the correlation of the LLG.
%We now compare the prediction of the SM \Eq{3p-SM-U-Scal} with numerical data of the LLg.    
We adopt different settings of time composition, in these settings, all data of different cases of correlation function sit on the same curve, see \Fig{3-point-data}.
%%%%   
\subsubsection{Data analysis}
We have obtained a sufficient amount of numerically estimated data concerning the correlation function, exploring various relationships among the three time variables. 
%These relationships include
%We have enough numerically estimated data of the  correlation function
%We amply data 
%for different relationships between the three times. The different settings between times are 
In all these functional relationships between the times defined in section \ref{subsec:LLg_correlations}, we find that the position auto-correlation function of the SM and the FND followed the dependence on \Eqs{U_3pt_corr} or \eq{U_3pt_corr_FnD}, 
%when one measure $r$  in units of $r_0$, time in units of $r_0/v$, 
and adopted the mapping of parameters $\alpha,\,\mu$ and $\xi$ (\cf~\Eq{llgpdf}).
% provided by the requesting $\alpha=\mu=2-\gamma\,.$ 
%For $h_2<t_2$ it scales like third moment of displacement when normalize by $t^{\gamma+1}$, \Eq{DeltaMoments} for $p=3$.
This is demonstrated in \FIG{3-point-data}, for the data collapse for one parameter dependent $3$-point position auto-correlation function with different functional relationships between three times. For 
$h_2<t_2$, we observe an excellent match between the LLG data and quantitative prediction of the SM, \Eq{U_3pt_corr} and the FND, \Eq{U_3pt_corr_FnD} at least for small values of $\xi$. 
%For $h_2<t_2$ this correlation scales as the third moment of displacement of the FnD dynamics \Eq{FnD-moment-scaling} with $p=3$.
% and requesting $\xi = 2 - \gamma$. 
For $h_2>t_2$, there is a different scaling and the agreement becomes gradually worse as $\xi$ increases. The three times $t_1$, $t_2$, and $t_3$ are far separated for the asymptotic scaling of small $h_2/t_2$. 
%When the three times are large enough one commonly observes that the correlation function decays like~$1/h$ as the prediction suggests in \Eq{U_3pt_corr}. 

%In the inset of both panels of \Fig{3-point-data}, we show the relative deviation between the numerical data and the theoretical prediction.  For $\tau_2<t_2$ the numerical data tend to be systematically scattered around the theoretical prediction as $\xi$ increases. The effect is small, only about a few percent for $\xi = 0.1$, but it grows for increasing $\xi$ ---reaching by some factor for $\xi = 0.6$. Even in the latter case, we consider the agreement to be excellent, however, because even for $\tau_2<t_2$ our data cover very few orders of magnitude on both axes. In this case, the relative deviation of the data from the prediction are $-0.058 \pm 0.015$ for $\xi = 0.1$ and $-0.47 \pm 0.12$ for $\xi = 0.6$. Moreover, the agreement for $\tau_2>t_2$ is perfect for increasing $\xi$.
%%%%%%%%%%%%%%%%%%%%%%%%%%%%%%%%%%%%%%%%%%%%%%%%%%%%%%%%%%%%%%%%%%%%%% 
%%%%%%%%%%%%%%%%%%%%%%%%%%%%%%%%%%%%%%%%%%%%%%%%%
%%%%%%%%%%%%%%%%%%%%%%%%%%%%%%%%%%%%%%%%%%%%%%%%%
\begin{figure}[h]%{figure}
  \centering
    \includegraphics[width=9.2cm, height=6.2cm]{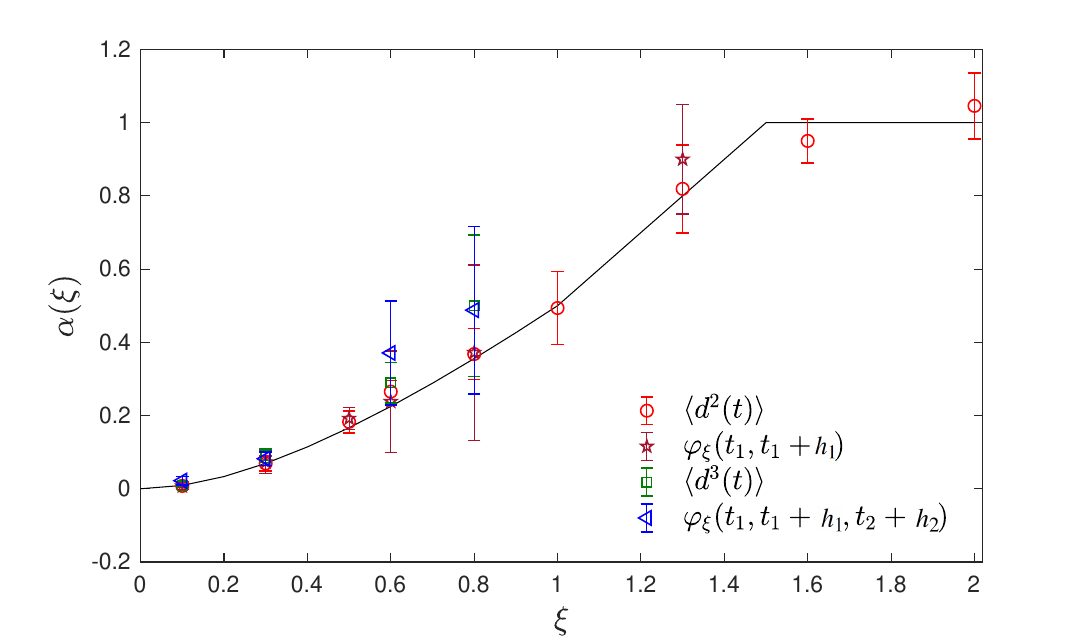}
     \caption{\label{fig:expnt-data}(Color online) This figure represents the parameters $(\xi;\alpha,\mu)$ functional relationship, \Eq{llgpdf} of the SM (or so FND) and the LLG, along with fitted values for some $\xi$. The fitted values with their bounds as a function of $\xi$ for the position moments and the correlations of order $n=2$ and $n=3$
are obtained by the best fit to the data and adopting
($\alpha(\xi)=n-best fit$).}
    \end{figure}
%%%%%%%%%%%%%%%%%%%%%%%%%%%%%%%%%%%%%%%%%%%%%%%%%
%%%%%%%%%%%%%%%%%%%%%%%%%%%%%%%%%%%%%%%%%%%%%%%%%
%%%%%%%%%%%%%%%%%%%%%%%%%%%%%%%%%%%%%%%%%%%%%%%%%

In \FIG{expnt-data}, we represent the theoretical mapping of the SM (or so FND) and the LLG by their $(\xi;\alpha,\mu)$ relation shown in \Eq{llgpdf}. We also represent the fitted values for some $\xi$ along the curve of the position moments and correlations. For the position moments data, which are obtained by numerically estimating the power law exponents $\eta$ and $\overline{\eta}$ of the second moment $\langle d^2(t)\rangle$ and third moment $\langle d^3(t)\rangle$, respectively, and adopting $\mu(\xi) = \alpha(\xi) = 2 - \eta(\xi) \simeq 3 - \overline{\eta}(\xi)$.
%; $\gamma_c$ and $\omega_c$ of the $2$-, and $3$-point position auto-correlation while requesting $\alpha(\xi)=2-\gamma_c(\xi)$ and $\alpha(\xi)=3-\omega_c(\xi)$ respectively as a function of $\xi$.
 %, then by adopting $\alpha(\xi)=2-\gamma(\xi)$.
%In \FIG{expnt-data}, we estimate the power law exponents $\gamma$ as a function of $\xi$, for the second and third moments of displacement, as well as for the $2$-, and $3$-point position auto-correlation function. 
Likewise, the fitted values of the position correlations $\varphi_\xi$ are obtained by estimated the power law exponent of \EQ{np-corr-momnt}, $\omega_p$, with $p=2,3$, where the exponents $\omega_2$ and $\omega_3$ render the $2$- and $3$-point position correlation functions, respectively, and adopting $\mu(\xi) = \alpha(\xi) = 2 - \omega_2 \simeq 3 - \omega_3$ (\cf~\Eqs{sam-sclng-momnt-corr-SM} and \eq{sam-sclng-momnt-corr-FnD}). 
%The time composition in the correlations is either $\tau=~const.$ or $\tau\sim t_1^q$ with $q<1$. 
The power law scaling of the position auto-correlation function does the same as the asymptotic scaling of the position moments amounts with single time~$t$ (see remark~\ref{rem:mom-corr-SM}). These findings confirm the prediction of exponents provided by \Eq{MOMBUR}. Hence they can be used for the higher order position moments and the correlation functions.

\section{Discussion and Conclusion\label{sec:concl}}
The investigation on the equivalence of observables between the SM and the LLG was started by Salari et al.~\cite{Salari}, who observed that the moments scale in the same fashion. Since moments contain only partial information on the transport systems, knowledge of the correlation is an essential ingredient to highly characterize anomalous transport dynamics~\cite{Sokolov}. Therefore, Giberti et al.~\cite{GRTV17} derived the $2$-point position auto-correlation function of the SM in several scaling forms and compared it with the numerically estimated position auto-correlation function of the LLG. They found a remarkable agreement in scaling, at least for lower scatterer density (\ie for small values of $\xi$). Findings on the coincidence of the position moments and the $2$-point position auto-correlation function do not hallmark the indistinguishability of these transport systems. Other observables are needed to distinguish these processes; for instance, the velocity auto-correlation functions of these processes are quite different.
%but this issue goes beyond the extent of the present work. It will be explored in a future paper.

In this paper, the general order position auto-correlation functions $\phi_\alpha$ of the SM \Eq{N-point-SM-summ} and $\rho_\mu$ of the FND \Eq{M-point-FnD-summ} are analytically computed. For a special case, their $3$-point position auto-correlation functions are also presented. Based on these analytical expressions, a single scaling relation for the $3$-point position auto-correlation function is established, allowing the representation of the correlation function in a scaling form where it only depends on the ratio of times $h_2/t_2$ (\cf~\Eqs{U_3pt_corr} and \eq{U_3pt_corr_FnD}). The excellent agreement between the numerical data of the LLG and the predictions obtained by the SM and the FND (symbols and dashed lines in \FIG{3-point-data}) establishes a new way to analyze correlations in anomalous transport. Additionally, it is argued that the position moments and correlations are posed in the same way, provided that the $(\xi;\alpha,\mu)$ relation follows \Eq{llgpdf}, as represented in Figure~\ref{fig:expnt-data}. It only depends on the exponent $\eta$ characterizing the MSD and the pre-factor of that asymptotic power law.

To conclude, at the very least, for small $\xi$, the $3$-point position auto-correlation function of the SM and the FND can capture the main features of the correlation function for the non-trivial anomalous transport process. We argued that systems with different microscopic dynamics but enjoy the same transport properties, such as position moments and correlation functions up to order 3, are considered. Consequently, for super-diffusive transport, the position moments and auto-correlation functions of the SM, the FND, and the LLG are dominated by ballistic trajectories.
The behavior of rare events at large distances in the SM, the FND, and the LLG is determined by the same physical origin: a single ballistic jump. This is described in the general framework of a single big jump in \cite{VBB19, VBB20}. This ballistic jump determines all the dynamical correlations when the diffusive parameters of these systems are small enough. Conversely, at short distances, the typical dynamical evolution differs significantly.
The big jump principle holds not only for power-law distributions with small exponents but also in models characterized by subexponential distributions, as noted in \cite{BV20}. However, when parameters such as \(\alpha\), \(\mu\), or \(\xi\) change, the big jump observed in different observables. For instance, the big jump determines the correlation functions specifically for a power-law probability density function with a sufficiently small exponent. 
%This corresponds to the scenario where dynamical correlations are influenced by rare, large deviations, aligning with the observations stated earlier.
%\textcolor{red}{The behavior of rare events at large distances in the SM, the FND, and the LLG is determined by the same physical origin, namely a single ballistic jump. This is described in the general framework of a single big jump in \cite{VBB19, VBB20}. This ballistic jump determines all the dynamical correlations when the diffusive parameters of these systems are small enough. Conversely, at short distances, the typical dynamical evolution differs significantly.} 
In the SM and the FND, the typical asymptotic dynamics are frozen; the trajectories become periodic within their neighboring cells in the SM and die in the FND. In contrast, for the LLG, the typical asymptotic dynamics can be diffusive or subdiffusive, depending on the value of $\xi$ \cite{VBB20}.
It is conjectured that the position moments and the auto-correlation function apply to a wide class of such systems \cite{VRTGM21}. Even with entirely different microscopic dynamics, the models agree regarding the characteristics of the displacement. However, the moments and correlations of the velocities may differ. In summary, these concepts have potential applications in various fields, ranging from dynamical systems and ecology to statistical physics, providing valuable insights into the behavior of complex systems.
%%%%%%%%%%%%%%%%%%%%%%%%%%%%%%%%%%%%%%%%%%%%%%%%%
\begin{acknowledgments}
M.T. gratefully acknowledges computational
resources provided by HPC@POLITO, the project
for Academic Computing of the Department of Control and Computer Engineering at the Politecnico di Torino, Italy \url{(http://hpc.polito.it)}.  M.T. also thanks to L. Rondoni for the useful discussions.
\end{acknowledgments}

\appendix

\section{Supporting derivation for the position moments}
\subsection{Derivation for $p^{th}$ position moments of the SM\label{app:pth-momnt}}
In this appendix we show the derivation of $p^{th}$ position moments, which is represented in Lemma~\ref{lem:pmoments}. The position moments as function of the number $n$ of iterations of the map $S_\alpha$. For $n\gg 2^{1/\alpha}$ one obtains
\begin{align}\label{eq:App-SM-p-intgrl}
\left\langle (x_n - x_0)^p \right\rangle  
 \simeq &\, 2 \,\int\limits_0^{\ell_n} \diff x \, n^p       + 2 \,\int\limits_{\ell_n}^{1/2} \diff x \, \left( x^{-1/\alpha} - 2^{1/\alpha} \right)^p\,, \\
  %\left\langle (x_n - x_0)^p \right\rangle
  \simeq &
  \,2\, n^p  \:  \ell_n 
    + \frac{2}{1 - p/\alpha} \: ( 2^{-1+p/\alpha} - \ell_n^{1-p/\alpha} ) 
    + \mathcal{O}(1)\,,
    \nonumber\\
   \sim &
    \frac{2 \,p}{p - \alpha} \: n^{p - \alpha} + \mathcal{O}(1)\,,
  \nonumber\\[2mm]
   \sim & \left\{ \begin{array}{ll}
                   const.\,,       & 
                   \text{ for }  p < \alpha \, ,
%                   \\[2mm]
%                   \displaystyle 2 \; \ln\frac{n^\alpha}{2}\,,         & \text{ for }  p = \alpha \, ,
                   \\[2mm]
                   \displaystyle \frac{2\, p}{p-\alpha} \, n^{p-\alpha} \,,  & \text{ for }  p > \alpha > 0 \,, \label{eq:App-SM-p-momnt}
    \end{array}\right.
\end{align}
while $p=\alpha$, \Eq{App-SM-p-intgrl} leads to
\begin{equation}\label{eq:App-Sm-a-momnt}
\left\langle (x_n - x_0)^\alpha \right\rangle
  \sim                  \displaystyle 2 \; \ln\frac{n^\alpha}{2}  \, .
\end{equation}
Collecting terms from \Eqs{App-SM-p-momnt} and \eq{App-Sm-a-momnt}, completes the proof of Lemma~\ref{lem:pmoments}.

%%%%%%%%%%%%%%%%%%%%%%%%%%%%%%%%%%%%%%%%%%%%%%%%%%%%%%%%%%%%%%%%%%%%%%%%%%%%%%%%%%%%%%%%%%%%%%%%%%%%

\subsection{Derivation for $p^{th}$ position moments of the FND \label{app:pth-momnt-FnD}}
This appendix shows the derivation of $p^{th}$ moments of the FND, the asysmptotic scaling is represented in Lemma~\ref{lemma:FnD-momnts}. For $p=\mu$, the $p^{th}$ position moments can be obtained
\begin{align}
\langle |\Delta x(t)|^p \rangle =& \langle |x(x_0,t) - x_0|^p \rangle \,, \nonumber \\
=& \int\limits_{0}^{1} |x(x_0,t) - x_0|^p dx_0 \,, \nonumber\\
=& \int\limits_{0}^{P(>t)} t^\mu \,dx_0
+  \int\limits_{P(>t)}^{1} \left(t_c(x_0)\right)^\mu dx_0\,, \label{eq:app_momnts-FnD}
\end{align}
where $t_c(x_0)$ is the final position of the particle and has power law tail, expressed in \Eq{FnD_tail}. The probability $P(>t)$ to perform a flight longer than $t$ amounts to the fraction of initial condition $x_0$ with $t_c(x_0)>t$ (\cf~\EQ{FnD_Probability}), such that from \Eqs{FnD_tail} and \eq{FnD_Probability}, we can write \Eq{app_momnts-FnD} as following 
\begin{align}
\langle |\Delta x(t)|^p \rangle &= \int\limits_{0}^{l/t^\mu} t^\mu \,dx_0
+  \int\limits_{l/t^\mu}^{1} \left(t_c(x_0)\right)^\mu dx_0\,, \\
&= t^p\frac{l}{t^\mu} + \frac{l^{p/\mu}}{1-p/\mu} 
\left(1 - \left(\frac{l}{t^\mu} \right)^{1-p/\mu} \right)\nonumber.
\end{align}
Rearrange and collect terms for the $t^{p-\mu}$ and $l^{p/\mu}$, one find
\begin{align*}
\langle |\Delta x(t)|^p \rangle = \frac{pl}{p-\mu}t^{p-\mu} + \frac{\mu}{\mu - p} l^{p/\mu}\,.
\end{align*}
Analogous derivation for $p=\mu$, and in the limit of long times $t>l^{1/\mu}$ completes the proof of Lemma~\ref{lemma:FnD-momnts}.

\bibliographystyle{unsrt} 
\bibliography{./anomalous}

\begin{thebibliography}{10}

\bibitem{SM75}
H.~Scher and E.~W. Montroll.
\newblock Anomalous transit-time dispersion in amorphous solids.
\newblock {\em Phys. Rev. B}, 12:2455--2477, (1975).

\bibitem{JBR08}
O.~G. Jepps, C.~Bianca, and L.~Rondoni.
\newblock Onset of diffusive behavior in confined transport systems.
\newblock {\em Chaos}.

\bibitem{JeRo06}
O.~G. Jepps and L.~Rondoni.
\newblock Thermodynamics and complexity of simple transport phenomena.
\newblock {\em J. Phys. A: Math. Gen.}

\bibitem{LZRVCMSE17}
E.~K. Lenzi, R.~S. Zola, H.~V. Ribeiro, D.~S. Vieira, F.~Ciuchi, A.~Mazzulla,
  N.~Scaramuzza, and L.~R. Evangelista.
\newblock Ion motion in electrolytic cells: Anomalous diffusion evidences.
\newblock {\em J. Phys. Chem. B}, 121:2882, (2017).

\bibitem{BGM12}
E.~Barkai, Y.~Garini, and R.~Metzler.
\newblock Strange kinetics of single molecules in living cells.
\newblock {\em Phys. Today}, 65:29, (2012).

\bibitem{SBAD12}
Y.~Sagi, M.~Brook, I.~Almog, , and N.~Davidson.
\newblock Observation of anomalous diffusion and fractional self-similarity in
  one dimension.
\newblock {\em Phys. Rev. Lett.}, 108:093002, (2012).

\bibitem{HA02}
S.~Havlin and D.~Ben-Avraham.
\newblock Diffusion in disordered media.
\newblock {\em Advances in Physics}, 51(1):187, (2002).

\bibitem{JLOM13}
J.-H. Jeon, N.~Leijnse, L.~B. Oddershede, and R.~Metzler.
\newblock Anomalous diffusion and power-law relaxation of the time averaged
  mean squared displacement in worm-like micellar solutions.
\newblock {\em New Journal of Physics}, 15(4):045011, (2013).

\bibitem{JMJM12}
J.-H. Jeon, H.~M.-S. Monne, M.~Javanainen, and R.~Metzler.
\newblock Anomalous diffusion of phospholipids and cholesterols in a lipid
  bilayer and its origins.
\newblock {\em Phys. Rev. Lett.}, 109:188103, (2012).

\bibitem{SW09}
J.~Szymanski and M.~Weiss.
\newblock Elucidating the origin of anomalous diffusion in crowded fluids.
\newblock {\em Phys. Rev. Lett.}, 103:038102, (2009).

\bibitem{GW10}
N.~Gal and D.~Weihs.
\newblock Experimental evidence of strong anomalous diffusion in living cells.
\newblock {\em Phys. Rev. E}, 81(1):020903(R), (2010).

\bibitem{FSBZ15}
D.~Froemberg, M.~Schmiedeberg, E.~Barkai, and V.~Zaburdaev.
\newblock A fractional dynamics approach.
\newblock {\em Phys. Rev. E}.

\bibitem{FVGB17}
R.~M. Feliczaki, E.~Vicentini, and P.~P. Gonz\'{a}lez-Borrero.
\newblock Dynamical transition on the periodic lorentz gas: Stochastic and
  deterministic approaches.
\newblock {\em Phys. Rev. E}.

\bibitem{KRS08}
R.~Klages, G.~Radons, and I.~M. Sokolov, editors.
\newblock {\em Anomalous Transport: Foundations and Applications}.
\newblock Wiley-VCH.

\bibitem{WCMS22}
W.~Wang, A.~G. Cherstvy, R.~Metzler, and I.~M. Sokolov.
\newblock Restoring ergodicity of stochastically reset anomalous-diffusion
  processes.
\newblock {\em Phys. Rev. Res.}, 4:013161, (2022).

\bibitem{MK00}
R.~Metzler and J.~Klafter.
\newblock The random walk's guide to anomalous diffusion: a fractional dynamics
  approach.
\newblock {\em Physics Reports}, 339(1):1--77, (2000).

\bibitem{MK04}
R.~Metzler and J.~Klafter.
\newblock The restaurant at the end of the random walk: recent developments in
  the description of anomalous transport by fractional dynamics.
\newblock {\em Journal of Physics A: Mathematical and General}, 37(31):R161,
  (2004).

\bibitem{Kla06}
R.~Klages.
\newblock {\em Microscopic chaos, fractals and transport in nonequilibrium
  statistical mechanics}.
\newblock World Scientific.

\bibitem{MJCB14}
R.~Metzler, J.-H.~Jeon ~, A.~G. Cherstvy, and E.~Barkai.
\newblock Anomalous diffusion models and their properties: non-stationarity{,}
  non-ergodicity{,} and ageing at the centenary of single particle tracking.
\newblock {\em Phys. Chem. Chem. Phys.}, 16:24128--24164, (2014).

\bibitem{Gas05}
P.~Gaspard.
\newblock {\em Chaos, scattering and statistical mechanics}.
\newblock Cambridge University Press.

\bibitem{D99}
J.~R. Dorfman.
\newblock {\em An introduction to chaos in non-equilibrium statistical
  mechanics}.
\newblock Cambridge University Press, Cambridge, (1999).

\bibitem{Kla96}
R.~Klages.
\newblock {\em Deterministic diffusion in one-dimensional chaotic dynamical
  systems}.
\newblock Wissenschaft \& Technik-Verlag, Berlin, (1996).

\bibitem{JBS03}
O.~G. Jepps, S.~K. Bathia, and D.J. Searles.
\newblock Wall mediated transport in confined spaces: exact theory for low
  density.
\newblock {\em Phys.\ Rev.\ Lett.}

\bibitem{Salari}
L.~Salari, L.~Rondoni, C.~Giberti, and R.~Klages.
\newblock A simple non-chaotic map generating subdiffusive, diffusive, and
  superdiffusive dynamics.
\newblock {\em Chaos}.

\bibitem{Zas02}
G.~M. Zaslavsky.
\newblock Chaos, fractional kinetics, and anomalous transport.
\newblock {\em Phys. Rep.}

\bibitem{DC00}
C.~P. Dettmann and E.~G.~D. Cohen.
\newblock Microscopic chaos and diffusion.
\newblock {\em J. Stat. Phys.}

\bibitem{CFVN02}
F.~Cecconi, {D.~del} Castillo-Negrete, M.~Falcioni, and A.~Vulpiani.
\newblock The origin of diffusion: The case of non chaotic systems.
\newblock {\em Physica D}.

\bibitem{GNZ85}
T.~Geisel, J.~Nierwetberg, and A.~Zacherl.
\newblock Accelerated diffusion in josephson junctions and related chaotic
  systems.
\newblock {\em Phys. Rev. Lett.}, 54:616--619, (1985).

\bibitem{BF99}
E.~Barkai and V.~Fleurov.
\newblock Stochastic one-dimensional {L}orentz gas on a lattice.
\newblock {\em J. Stat. Phys.}, 96:325, (1999).

\bibitem{Sokolov}
I.~M. Sokolov.
\newblock Models of anomalous diffusion in crowded environments.
\newblock {\em Soft Matter}, 8:9043, (2012).

\bibitem{Z07}
G.~Zaslavsky.
\newblock {\em The physics of Chaos in Hamiltonian systems, second edition}.
\newblock (2007).

\bibitem{DKU03}
S.~Denisov, J.~Klafter, and M.~Urbakh.
\newblock Dynamical heat channels.
\newblock {\em Phys. Rev. Lett.}

\bibitem{LWWZ05}
B.~Li, J.~Wang, L.~Wang, and G.~Zhang.
\newblock Anomalous heat conduction and anomalous diffusion in nonlinear
  lattices, single walled nanotubes, and billiard gas channels.
\newblock {\em Chaos}.

\bibitem{BFK00}
E.~Barkai, V.~Fleurov, and J.~Klafter.
\newblock One-dimensional stochastic {L}{\'e}vy-{L}orentz gas.
\newblock {\em Phys. Rev. E}, 61:1164, (2000).

\bibitem{BCV10}
R.~Burioni, L.~Caniparoli, and A.~Vezzani.
\newblock {L{\'e}vy} walks and scaling in quenched disordered media.
\newblock {\em Phys. Rev. E}.

\bibitem{LWMC23}
Y.~Liang, W.~Wang, R.~Metzler, and A.~G. Cherstvy.
\newblock Anomalous diffusion, nonergodicity, non-gaussianity, and aging of
  fractional brownian motion with nonlinear clocks.
\newblock {\em Phys. Rev. E}, 108:034113, (2023).

\bibitem{WCLM20}
W.~Wang, A.~G. Cherstvy, X.~Liu, and R.~Metzler.
\newblock Anomalous diffusion and nonergodicity for heterogeneous diffusion
  processes with fractional gaussian noise.
\newblock {\em Phys. Rev. E}, 102:012146, (2020).

\bibitem{Z08}
V.~Y. Zaburdaev.
\newblock Microscopic approach to random walks.
\newblock {\em J. Stat. Phys.}, 133(1):159, (2008).

\bibitem{BF07}
A.~Baule and R.~Friedrich.
\newblock A fractional diffusion equation for two-point probability
  distributions of a continuous-time random walk.
\newblock {\em Europhysics Letters ({EPL})}, 77(1):10002, (2007).

\bibitem{BS07}
E.~Barkai and I.~M. Sokolov.
\newblock Multi-point distribution function for the continuous time random
  walk.
\newblock {\em J. Stat. Mech. Theory Exp}, 2007(08):P08001, (2007).

\bibitem{ZDK15}
V.~Zaburdaev, S.~Denisov, and J.~Klafter.
\newblock {L}\'evy walks.
\newblock {\em Rev. Mod. Phys.}, 87:483, (2015).

\bibitem{VRTGM21}
J.~Vollmer, L.~Rondoni, M.~Tayyab, C.~Giberti, and C.~M.-Monasterio.
\newblock Displacement autocorrelation functions for strong anomalous
  diffusion: A scaling form, universal behavior, and corrections to scaling.
\newblock {\em Phys. Rev. Res.}, 3:013067, (2021).

\bibitem{CMGV99}
P.~Castiglione, A.~Mazzino, P.~Muratore-Ginanneschi, and A.~Vulpiani.
\newblock On strong anomalous diffusion.
\newblock {\em Physica D: Nonlinear Phenomena}, 134(1):75--93, (1999).

\bibitem{VBB19}
A.~Vezzani, E.~Barkai, and R.~Burioni.
\newblock Single-big-jump principle in physical modeling.
\newblock {\em Phys. Rev. E}, 100:012108, (2019).

\bibitem{VBB20}
A.~Vezzani, E.~Barkai, and R.~Burioni.
\newblock Rare events in generalized {L}évy {W}alks and the {B}ig {J}ump
  principle.
\newblock {\em Sci Rep}, 10:2372, (2020).

\bibitem{BCLL16}
A.~Bianchi, G.~Cristadoro, M.~Lenci, and M.\ Ligab\`o.
\newblock Random walks in a one-dimensional {L\'evy} random environment.
\newblock {\em J Stat Phys}, 163:22 -- 40, (2016).

\bibitem{Z23}
M.~Zamparo.
\newblock {Large fluctuations and transport properties of the L\'evy–Lorentz
  gas}.
\newblock {\em Annales de l'Institut Henri Poincaré, Probabilités et
  Statistiques}, 59(2):621 -- 661, (2023).

\bibitem{GRTV17}
C.~Giberti, L.~Rondoni, M.~Tayyab, and J.~Vollmer.
\newblock Equivalence of position–position auto-correlations in the slicer
  map and the {L}\'evy–{L}orentz gas.
\newblock {\em Nonlinearity}, 32(6):2302, (2019).

\bibitem{GSSPCM18}
D.~M.-Garcia, T.~Sandev, H.~Safdari, G.~Pagnini, A.~Chechkin, and R.~Metzler.
\newblock Crossover from anomalous to normal diffusion: truncated power-law
  noise correlations and applications to dynamics in lipid bilayers.
\newblock {\em New Journal of Physics}, 20(10):103027, (2018).

\bibitem{BV20}
R.~Burioni and A.~Vezzani.
\newblock Rare events in stochastic processes with sub-exponential
  distributions and the big jump principle.
\newblock {\em Journal of Statistical Mechanics: Theory and Experiment},
  2020(3):034005, (2020).

\end{thebibliography}

\end{document}